\documentclass{article}
 
\usepackage[english]{babel}
\usepackage{amsfonts,amsmath,dsfont}
\usepackage{amsthm}
\newcommand{\imag}{\mathrm{i}}
\newcommand{\I}{\mathrm{i}}
\newcommand{\ez}{\mathrm{e}}
\newcommand{\df}{\mathrm{d}}
\newcommand{\hsc}{h}
\newcommand{\End}{\mathrm{End}}

\renewcommand{\exp}[1]{\mathrm{e}^{#1}}
 \newcommand{ \tr }{ \mathrm{tr} }

\newcommand{\hfast}{\mathcal{H}_{\mathrm{f}}}
\newcommand{\Or}{{\mathcal{O}}}
\newcommand{\R}{{\mathbb{R}}}
\newcommand{\C}{{\mathbb{C}}}
\newcommand{\N}{{\mathbb{N}}}
\newcommand{\Z}{{\mathbb{Z}}}
\newcommand{\D}{{\mathrm{d}}}
\newcommand{\Hi}{ \mathcal{H} }
\newcommand{\epsi}{ \varepsilon}

\newtheorem{theo}{Theorem}
\newtheorem{prop}{Proposition}
\theoremstyle{definition}

\theoremstyle{plain}
\newtheorem{lemma}{Lemma}
\newtheorem{coro}{Corollary}
\theoremstyle{rema}

\title{Semiclassical approximations for   Hamiltonians  with operator-valued symbols\thanks{This work was supported by the German Science Foundation (DFG) with the SFB-TR 71 and by the German-Israeli Foundation (GIF)}}
\author{Hans-Michael Stiepan, Stefan Teufel\\[1mm]
Mathematisches Institut, Universit\"at T\"ubingen, Germany.
}

\begin{document}

\maketitle

\begin{abstract}
We consider the semiclassical limit of quantum systems with a Hamiltonian given by the Weyl quantization of an operator valued symbol. Systems  composed of {\emph{slow}} and {\emph{fast}} degrees of freedom are of this form. Typically a small dimensionless parameter $\epsi\ll 1$ controls the separation of time scales 
and the limit $\epsi\to 0$ corresponds to an adiabatic limit, in which   the slow and fast degrees of freedom decouple. At the same time  $\epsi\to 0$ is the  semiclassical limit for the slow degrees of freedom. In this paper we show that the $\epsi$-dependent classical flow  for the slow degrees of freedom first discovered by Littlejohn and Flynn \cite{MR1133847},
 coming from an $\epsi$-dependent classical Hamilton function and an $\epsi$-dependent symplectic form,  
 has a concrete mathematical and physical meaning:
 Based on this flow we prove a formula for equilibrium  expectations, an Egorov theorem 
and transport of Wigner functions, thereby  approximating properties of the quantum system up to errors of order $\epsi^2$. In the context of Bloch electrons formal use of  this classical system has 
triggered considerable progress in solid state physics~\cite{RevModPhys.82.1959}.  
Hence we discuss in some detail the application of the general results to the Hofstadter model, which describes 
a two-dimensional gas of non-interacting electrons in a constant magnetic field in the tight-binding approximation.
 \end{abstract}

\newpage

\tableofcontents

\section{Introduction}

Semiclassical approximations play an important role in the understanding of many quantum mechanical phenomena. The standard situation is a Hamiltonian $\hat H = H(x,-\imag \hbar \nabla_x)$, acting on $L^2(\R^n)$, that is the Weyl quantization of a real valued function $H(q,p)$ on the classical phase space $\R^{2n}$. The goal of semiclassical methods is to approximate properties related to the quantum mechanical operator  $\hat H$ on $L^2(\R^n)$ in the limit $\hbar \to 0$ using the  classical Hamiltonian   $H$ on $\R^{2n}$ and its flow.
Examples of such properties are the spectrum and eigenfunctions of $\hat H$, special solutions of the time-dependent Schr\"odinger equation or the full   unitary group $\ez^{-\imag \hat H \frac{t}{\hbar}}$ or 
statistical expectation values $\tr(f(\hat H) \hat a)$. 

In many situations  only some physical degrees of freedom behave semiclassically, the simplest situation being particles with spin. Then the quantum mechanical state space is $L^2(\R^n,\C^L)$, or more general, $  L^2(\R^n,\hfast):= \Hi$, where $\hfast$ is the state space of the ``fast'' or ``fibre'' degrees of freedom.
Then the Hamiltonian is often of the form $\hat H = H(x,-\imag \epsi\nabla_x)$, where now $H:\R^{2n}\to \mathcal{L}(\hfast)$ is a function on classical phase space taking values in the linear self-adjoint operators on $\hfast$. The physical meaning of the small parameter $\epsi\ll 1$ depends on the concrete problem.  Also in this setting semiclassical approximations turned out very successful and a vast literature exists. 
The main content of this work are two results which, to our knowledge, have been neither proved nor claimed in this generality, but have been used already quite successfully in concrete problems, in particular in solid state physics, see \cite{RevModPhys.82.1959} and references therein.
 
We now explain our main results, postponing technical details to Section~2 and beyond. 
The first question is, which classical Hamiltonian system should be related to an operator-valued function $H(q,p)$ on classical phase space? We give the following answer: to each isolated non-degenerate eigenvalue $e_0(q,p)$ of $H(q,p)$ depending smoothly on $(q,p)$ we associate an almost invariant subspace  $\Pi^\epsi L^2(\R^n,\hfast)$, $\Pi^\epsi$ an orthogonal projection depending on $\epsi$, and a specific classical Hamiltonian system.
The restriction of the  Hamiltonian $\hat H$  to the range of $\Pi^\epsi$ is a generalized  adiabatic approximation  
which has been studied in great detail, see e.g.\ \cite{Teufel:2003} for an overview. 
The restricted Hamiltonian $\Pi^\epsi\hat H\Pi^\epsi$ is then amenable to semiclassical approximation using the mentioned classical Hamiltonian  system. At leading order the corresponding scalar Hamiltonian $h:\R^{2n}\to \R$ is just given by the eigenvalue itself, $h(q,p)= e_0(q,p)$, and this is basically well known. However, in many applications next to leading order effects play an important role and we show how to incorporate the first order corrections into a modified $\epsi$-dependent classical system.
The corrected Hamiltonian function takes the form
\begin{equation*} 
h =e_0+\epsi M,
\end{equation*}
where the function $M $ is given in terms of the spectral projection $\pi_0(q,p)$ of $H(q,p)$ corresponding to the eigenvalue $e_0(q,p)$.
In addition one needs to modify also the symplectic form on $\R^{2n}$,
  \begin{equation*} 
{\omega_\epsi} \;:=\; \omega_0 \,+\,\epsi \,\Omega\,,
\end{equation*}
where $\omega_0$ is the standard symplectic form on $\R^{2n}$.
With  $z=(q,p)$ and $\alpha,\beta =1,\ldots,2n$, the correction term $\Omega$ takes the form
\[
  \Omega_{\alpha\beta} \;:=\; -\,\imag \,\tr_{\hfast}\left(\pi_0[\partial_{z_\alpha}\pi_0,\partial_{z_\beta}\pi_0] \right) \,.
\]
As we will explain, $\Omega$ is the curvature form of the induced connection, often called the Berry connection,  on the line bundle over $\R^{2n}$ defined by the rank one projections $\pi_0(q,p)$.

The Hamiltonian equations of motion $\omega_\epsi(X_h,\cdot) = \D h$ imply that the  Hamiltonian vector-field $X_h$   is given  by 
$X_h^\alpha = -(\omega_\epsi)^{\alpha\beta}  \partial_\beta h$ and the classical equations of motion thus have the form
\begin{equation}\label{eq-motion}
\left(\begin{array}{c} \dot q\\\dot p \end{array}\right) \;=\;  \left( \begin{array} {cc} - \epsi \,\Omega^{pp} &          E_n+\epsi \Omega^{pq} \\    -E_n+\epsi \Omega^{qp} &- \epsi \,\Omega^{qq} \end{array} \right)  \left( \begin{array}{c} \partial_q \hsc \\ \partial_p \hsc     \end{array}   \right) \;+\;\Or(\epsi^2)\,.
\end{equation}
Here $E_n= {\rm diag}(1,\ldots, 1)$ denotes the $n\times n$ unit matrix and the error term results from the fact that 
the matrix $(\omega_\epsi)_{\alpha\beta}$ was inverted only approximately.
We denote the corresponding classical flow by $\phi^t_\epsi$ and  the Liouville measure induced by the symplectic form $\omega_\epsi $ by
  $\lambda_\epsi $.
  
We prove the following two statements about semiclassical approximations in terms of the classical system defined above.\\[3mm]
\noindent {\bf Stationary expectations} (cf.\ Theorem~\ref{satz:exp}) \\Let $f:\R\to \R$ be real valued and $a:\R^{2n}\to \R$ be a suitable scalar observable. Then 
\begin{equation}\label{claim1}
\tr_\Hi\left(\Pi^\epsi\, f( \hat H ) \,\hat a \right) \;=\; \frac{1}{(2\pi\epsi)^n} \left(\int_{\R^{2n}} \df \lambda_\epsi \; f(h(q,p))\, a(q,p) \;+\;\Or(\epsi^{2})\right)\,.
\end{equation} 
 Here one could think of $f$ being a distribution function like the Fermi-Dirac distribution and $\hat a$ some trace-class observable. A nice application of this formula where the leading order approximation
\[
\tr_\Hi\left( \Pi^\epsi\,f( \hat H ) \,\hat a \right) \;=\; \frac{1}{(2\pi\epsi)^n} \left( \int_{\R^{2n}} \df q\df p \; f(e_0(q,p))\, a(q,p) \;+\;\Or(\epsi^{1})\right)
\] 
would give the wrong  answer, namely zero,  
is the derivation of a formula for the orbital magnetization  in periodic media, c.f.\ \cite{gat2003semiclassical,RevModPhys.82.1959} and Corollary~\ref{MagCor} in Section~\ref{Hofstadter}.\\[3mm]
\noindent{\bf Egorov theorem} (cf.\ Theorem~\ref{satz:ego})\\ One can approximate the time evolution of quantum mechanical observables using the classical flow $\phi_\epsi^t$.
 Let again $a:\R^{2n}\to \R$ be a suitable scalar observable,
\[
A(t) := \ez^{\imag\hat H\frac{t}{\epsi}}\,\hat a\, \ez^{-\imag\hat H\frac{t}{\epsi}}
\]
its quantum mechanical Heisenberg time-evolution and $a(t) := a\circ\phi^t_\varepsilon$ its classical evolution. Then
\begin{equation}\label{claim2}
\left\| \Pi^\epsi\left( A(t) - \widehat{a(t)}\right)\Pi^\epsi\right\| \;= \;\Or(\epsi^2)
\end{equation}
and, as a consequence, for any state $\rho=\Pi\rho\Pi$ that is trace-class, 
\[
\tr_\Hi \big(\rho \,A(t) \big) \;= \;\tr_\Hi\big(\rho \,\widehat{a(t)}\big) \;+\;\Or(\epsi^2)\,.
\]

Let us try to give a  heuristic meaning to the modifications in the classical Hamiltonian system.
As we will see, the projection $\Pi^\epsi$ is an order $\epsi$ modification of $\hat\pi_0$. Thus the energy of the fast degrees of freedom  for states in the range of $\Pi^\epsi$ is not exactly $e_0$, but must be modified by $\epsi M$. This interpretation is very natural given the symmetric form of $M$ in (\ref{Malt}).
The correction to the geometry of phase space, expressed in the modified symplectic form $\omega_\epsi $ is necessary, since the wave functions in the range of $\Pi^\epsi=\hat\pi_0 +\Or(\epsi)$ take values, micro-locally, in a line bundle with nontrivial geometry defined by $\pi_0$.

The literature on semiclassical approximations for Hamiltonians with operator-valued symbols is huge and we give a short survey in Section~\ref{sec-lit}. At this point let us  mention what we believe are the most important contributions. 

The algebraic part  of the construction 
of the projection $\Pi^\epsi$ is a central step that was done by different approaches and independently by Helffer and Sj\"ostrand \cite{HS90} and  Emmrich and Weinstein \cite{MR1376438}. Emmrich and Weinstein  also showed that $h$ appears in the transport equation for the amplitude of the WKB approximation, but they did not make use of the Hamiltonian flow of~$h$. 
However,   the   $M$ term and the modified symplectic form $\omega_\epsi$ appeared already earlier in a work of Littlejohn and Flynn \cite{MR1133847}, who used the equations of motion~(\ref{eq-motion}) in order to construct Lagrangian submanifolds on which WKB wave functions are supported. 
But they arrive at these equations along a slightly different route and, as far as we understand their work, they considered them rather  a technical tool than the correct physical equations for the slow degrees of freedom. 

As the equations of motion  for the   position and quasi-momentum of a Bloch electron, a special case of  (\ref{eq-motion}) appeared for the first time explicitly  in a work of Chang and Niu \cite{chang-1995-53}. Since then the ``modified semiclassical model'' was at the basis of a lot of theoretical progress in solid state physics, see \cite{RevModPhys.82.1959} for a review. The only rigorous justification appeared in the case of Bloch electrons, where an Egorov theorem in the above explained sense was proven in \cite{PST03} and slightly generalized in \cite{GiuseppeMax}. 
 
After the success of the modified semiclassical model in solid state physics the literature concerned with its generalizations and justifications became huge. However, to our best knowledge the claims
(\ref{claim1}) and (\ref{claim2}) have not been clearly stated, nor rigorously or even just systematically derived until now, the one exception being \cite{PST03}. In addition to being new, we believe that our derivation of (\ref{claim1}) and (\ref{claim2}) is particularly simple and transparent, granted that one is willing to work with the Weyl calculus. 

A key feature that distinguishes our approach and \cite{MR1376438} from the majority of other approaches, starting with \cite{MR1133847} and including \cite{PST03},  is that we make no use of eigenfunctions $\varphi_0(q,p)$ of $H(q,p)$, but only of the projections $\pi_0(q,p)$. While the former are gauge-dependent, since with $\varphi_0(q,p)$ also $\ez^{\imag f(q,p)}\varphi_0(q,p)$ for  real-valued $f$ is an eigenfunction, the projections are unique. Moreover, in some applications like magnetic Bloch electrons, a global smooth choice for  $\varphi_0(q,p)$ simply does not exist  for geometric reasons. Exactly this problem was our motivation for developing the new approach presented here and we discuss the application to magnetic Bloch bands of the Hofstadter model in Section~\ref{Hofstadter}. The general situation of magnetic Bloch electrons is the content of a future paper \cite{ST:2011}. 

%
Let us stress that
in this work we make no serious attempt to achieve greatest possible generality. Instead we focus mostly  on structural aspects and on covering the application to the Hofstadter model discussed in Section~\ref{Hofstadter}. 
Therefore we will make  stronger assumptions than presumably necessary in order to avoid
distracting technicalities.
But  we will comment in the following on where and how one can relax the assumptions. 
Most works, including \cite{MR1133847,MR1376438}, are concerned with asymptotic expansions of formal symbols only and ignore the problem of turning the algebraic computation into statements about norm convergent approximations. 
Note that  these algebraic computations and thus the structure of the classical Hamiltonian system  remain exactly the same in all applications. Our restrictive assumptions come only into play in the steps where 
symbolic expansions are turned into statements about  operators and traces.
If one is willing to forego mathematical rigor in this step, one can   freely apply our results even in cases where the precise mathematical  assumptions are not satisfied. For an alternative but much more involved proof of the Egorov theorem with a slightly larger range of validity we refer to \cite{diss}.

We end the introduction with a short overview of the structure of this paper.
In Section~\ref{sec:sat} we first explain the precise setting and assumptions and   give a brief introduction to Weyl calculus.  Then we state the  known result concerning the existence and properties of the projection $\Pi^\epsi$ defining the almost invariant subspace related to the eigenvalue $e_0$.
In Section~\ref{sec-clas} we introduce the classical Hamiltonian $h$ and show how it is related to the restriction of $\hat H$ to the range of $\Pi^\epsi$.
Section~\ref{sec:sc} contains the mathematical statements and proofs of our main results (\ref{claim1}) and (\ref{claim2}). In Section~\ref{section:correm}   we 
prove that $\omega_\epsi$ is indeed a symplectic form and compute an explicit formula for the Liouville measure $\lambda_\epsi $. Then we show how $\Omega$ is related to the curvature of the Berry connection.
In Section~\ref{sec-lit} we discuss some related literature in view of our results. Finally, in Section~\ref{Hofstadter}, we discuss a tight-binding model for a two-dimensional gas of non-interacting particles in a constant magnetic field, the so called Hofstadter model. Applying our abstract results, we compute the free energy per volume at fixed temperature and chemical potential and   its derivative with respect to the magnetic field, the magnetization. Moreover, we also compute the Hall current at zero temperature when the Fermi level is in a gap, leading to the quantized conductivity 
in terms of Chern numbers. Although similar results have been established by other techniques, see \cite{SBT12}, we believe that the derivation of these expressions in terms of the semiclassical model is very simple and transparent, once the model is established.  \\[3mm]
%
{\bf Acknowledgements.} We thank   Giuseppe De Nittis, Omri Gat, Max Lein, Gian\-luca Panati and Hermann Schulz-Baldes     for   inspiring discussions and
 Yuri Kord\-yukov for pointing out several typos and glitches in an earlier version.  
S.T.\ is also grateful to Giuseppe De Nittis and Gian\-luca Panati for sending us their unpublished  notes on  the Harper and the Hofstadter model. Finally S.T.\ thanks Peter Vastag for his involvement in a joint project on the semiclassics of the Hofstadter model.

\section{The adiabatic approximation}
\label{sec:sat}

\subsection{General setting and precise assumptions}
\label{sec:sett}

We assume that the state space $\mathcal{H}$ decomposes as $\mathcal{H}=L^2(\mathbb{R}^n)\otimes\mathcal{H}_{\mathrm{f}}$, where the subscript 
 f stands for   ``fast'' degrees of freedom or for ``fibre''. 
Using Lebesgue-Bochner spaces we may rewrite this
as 
\begin{equation*}
\mathcal{H}=L^2(\R^n)\otimes\mathcal{H}_{\mathrm{f}} \cong L^2(\R^n,\mathcal{H}_{\mathrm{f}}).
\end{equation*}
One can replace $\R^n$ by a flat $n$-dimensional torus with only very minor changes to the following analysis. 
Most importantly, one has to use an appropriately modified pseudo-differential calculus as explained, e.g., in \cite{PST03}. We will give some more details on the modifications when we discuss the Hofstadter model in Section~\ref{Hofstadter}.
%
%
For manifolds with non-vanishing curvature the first order corrections cannot be implemented using a modified flow, see \cite{Jonas}.
\smallskip

\noindent{\bf Assumption 1}.
{\em On $\Hi$  we study a self-adjoint $\epsi$-dependent Hamiltonian $\hat{H}$  that   is given as the $\epsi$-Weyl quantization $\hat{H} =\mathrm{op}^W(H) =$ ``$ H(\epsi, x,-\imag \epsi \nabla_x)$''
of a semiclassical symbol }
\[
  H(z) = H_0(z) + \epsi H_1(z) + \xi\cdot z\quad  \mbox{ \em with }H_0, H_1 \in S^0(\mathcal{B}(\hfast))\mbox{ \em and } \xi\in \R^{2n}\,.
\]  
For any Banach space $X$ and $k\in\R$ let $S^k(X)$ denote the smooth functions  from $\R^{2n}$ to $X$ that are bounded together with all their derivatives by a constant times $\langle z\rangle^k:=(1+|z|^2)^\frac{k}{2}$, i.e.\ for $f\in S^k(X)$
\[
\|f\|_{k, r } := \max_{ \small \begin{array}{l} \alpha\in\N^{2n}_0\\ |\alpha|\leq r \end{array}}  \sup_{z\in\R^{2n} } \| \langle z\rangle^k\, \partial_z^\alpha   f(z)\|_{X} <\infty\qquad\mbox{for all $r\in\N_0$.}
\]
The family of seminorms $\|\cdot\|_{k, r }$ turns $S^k(X)$ into a Fr\'echet space. 
The space of  uniformly bounded functions $A: [0,\epsi_0)\to S^k(X)$ 
 is denoted by $S^k(\epsi,X)$.
 
 Note that compared to the introduction we allow for a subprincipal symbol $H_1$ and  explicitly 
separate a linear part in the symbol for technical reasons.
As will be recalled in the next section,  $\hat H_0$ and $\hat H_1$ are bounded operators. As a consequence
$\hat H$ is self-adjoint on the maximal domain of the operator $\xi_q\cdot x \,-\,\I\epsi \xi_p\cdot \nabla_x$ denoted by $D(\hat H)$ in the following. 
The assumption that the symbol $H$ is bounded up to a part  linear in $z$, allows to prove quite directly 
uniform statements in the following. However, all results are local in phase space $\R^{2n}$ and more general Hamiltonians, e.g.\ with quadratically growing symbols, can be brought, at least formally, into the required form by appropriate micro-localization in phase space.
A rigorous implementation of this idea in the general case is presumably quite technical, but has been done successfully e.g.\ in the context of the Born-Oppenheimer approximation, see \cite{MR2527490,TeWa12}.
 For the construction of $\Pi^\epsi$ it is also known that  one can allow for $H$ to take its values in the unbounded operators on $\hfast$, c.f.\ \cite{PST03}. We show in \cite{ST:2011} that also the other steps of our semiclassical approximation work in this case.\smallskip
 
 \noindent{\bf  Assumption 2}. {\em We assume that $H_0(q,p)$ has a non-degenerate eigenvalue $e_0(q,p)$ such that $e_0:\R^{2n}\to \R$ is continuous and satisfies the uniform gap condition}
 \[
 {\rm dist}( e_0(q,p) ,\, \sigma(H_0(q,p))\setminus\{e_0(q,p)\} )\geq g>0\,.
 \]  
 We denote the eigenprojection of $H_0(q,p)$ to the eigenvalue $e_0(q,p)$ by $\pi_0(q,p)$,
i.e.
 \[
 H_0(q,p) \,\pi_0(q,p) \;=\; e_0(q,p)\,\pi_0(q,p)\,.
 \]
 The gap condition implies smoothness of $e_0$ and of $\pi_0$ and is also crucial for the construction of super-adiabatic subspaces, cf.\ Proposition~\ref{satz:sat}.  It is known that at eigenvalue crossings  even a leading order approximation in terms of a classical flow becomes much more intricate, see e.g.\ \cite{MR1906196,MR2163659,DFJ09} and references therein.

 \subsection{Key formulas from Weyl calculus}\label{WeylForm}
 In this section we collect for later reference some key formulas and results of the Weyl calculus and introduce our notation.  Readers familiar with the Weyl calculus can skip this section. For more details on $\epsi$-pseudo-differential operators we refer e.g.\ to \cite{martinez2002introduction,MR1735654}. For a short summary on Weyl calculus with operator-valued symbols  the readers can also consult Appendix~A of \cite{Teufel:2003}.

For a Schwartz function $\psi \in \mathcal{S}(\mathbb{R}^n,\hfast)$  and $A\in S^k(\epsi, \mathcal{B}(\hfast))$ the action of  the Weyl quantization $\hat A = \mathrm{op}^W(A)$ of $A$ can be defined through the oscillatory integral 
\begin{equation}
\label{eq:wq}
(\hat{A}\psi)(x)=\frac{1}{(2\pi \epsi)^n}\int_{\mathbb{R}^{2n}}\df p \df y \,A\left(\epsi,\tfrac{1}{2}(x+y),p\right)\ez^{\imag p \cdot (x-y)/\epsi} \psi (y).
\end{equation}
For $k\geq 0$, by the Calderon-Vaillancourt theorem,    $\hat A$ can be extended to a bounded operator on $\Hi=L^2(\R^{n},\hfast)$. More precisely there exists a constant $c_n$ depending only on the dimension $n$, such that for all $\epsi\in[0,\epsi_0)$
\begin{equation*}
\lVert \hat{A}\rVert_{\mathcal{B}(\Hi)}\leq c_n 
 \|A(\epsi)\|_{0,2n+1}\,.
\end{equation*}
The composition of operators induces a composition of symbols. For any  $A \in S^{k_1}(\epsi, \mathcal{B}(\hfast))$  and $B \in S^{k_2}(\epsi, \mathcal{B}(\hfast))$, there exists a symbol $C\in S^{k_1+k_2}(\epsi, \mathcal{B}(\hfast))$ denoted by
$C= A\# B$ such that $\hat A\hat B= \hat C$. The bilinear map $\#: S^{k_1}\times S^{k_2}\to S^{k_1+k_2}$ is called the 
Moyal product and it
is continuous with respect to the Fr\'echet topologies uniformly in $\epsi$, i.e.\ for any $r\in\N_0$ there is a  $\tilde r \in \N_0$ and a
constant  $c_{r,\tilde r}<\infty$   such that 
\[
 \| (A\# B)(\epsi) \|_{k_1+k_2,r} \leq c_r \, \|A(\epsi)\|_{k_1,\tilde r}\|B(\epsi)\|_{k_2,\tilde r}
\]
for all $\epsi\in[0,\epsi_0)$. The last statement follows e.g.\ from inspecting the proof of Thm.~2.41 in \cite{Folland}.

If for $A\in S^k(\epsi, X)$ there exists a sequence $(A_j)_{j\in\N_0}$ in $S^k(\epsi,X)$ such that 
\[
\sup_{\epsi\in[0,\epsi_0) }\Big\| \epsi^{-(m+1)} \Big( A(\epsi) - \sum_{j=0}^m \epsi^j A_j(\epsi)\Big) \Big\|_{k,r } <\infty
\]
  for all $r \in \N_0 $ and $m\in\N_0$, then    one writes  
$ A \asymp \sum_{j=0}^\infty \epsi^j A_j$ in  $S^k(\epsi,X)$. 
If $A\in S^k(\epsi,X)$ has an asymptotic expansion with coefficients $A_j\in S^k(X)$ not depending on $\epsi$, then $A$ is called a classical symbol,
$A_0$ its principal symbol and  $A_1$ its subprincipal symbol. 

The Moyal product $C:= A\# B \in S^{k_1+k_2}(\epsi, \mathcal{B}(\hfast))$ of   symbols $A \in S^{k_1}(\epsi, \mathcal{B}(\hfast))$  and $B \in S^{k_2}(\epsi, \mathcal{B}(\hfast))$ has an explicit asymptotic expansion 
$ C \asymp \sum_{j=0}^\infty \epsi^j C_j$ in  $S^{k_1+k_2}(\epsi,X)$
such that the remainder maps 
\[
R_{m+1}: S^{k_1}\times S^{k_2}\to S^{k_1+k_2}\,, \quad (A,B)\mapsto R_{m+1}:= \epsi^{-(m+1)} \Big( C(\epsi) - \sum_{j=0}^m \epsi^j C_j(\epsi)\Big)  
\]
are continuous.
The expansion 
starts with the pointwise product $C_0 (\epsi)= A(\epsi)B(\epsi)$ and the Poisson bracket $C_1(\epsi) =  - \tfrac{ \imag}{2} \{ A (\epsi), B(\epsi) \}$, 
where
\[
\{ A , B \} \;:=\; \partial_p A \cdot \partial_q B   - \partial_q A \cdot \partial_p B 
:=\sum_{j=1}^n  \left( \partial_{p_j} A  \, \partial_{q_j} B   - \partial_{q_j} A \, \partial_{p_j} B \right)\,.
\]
  For classical symbols $A$ and $B$ the Moyal product $C:= A\# B$ is also a classical symbol with an asymptotic expansion starting with
\[
A\# B \;\asymp\;  A_0 B_0 \;+\; \epsi\left( A_1 B_0 + A_0B_1 - \tfrac{\imag}{2} \{ A_0, B_0\}
\right) \;+\;  \Or(\epsi^2)\,.
\]
%
%
%
%
%
%
Since $A $ and $B $ are operator valued functions, they do not commute in general and neither do their derivatives. Hence, in general, $\{A,A\}\not=0$, but, if the derivatives of $A$ are trace-class,
\begin{equation}\label{cyctrace}
\tr_{\hfast} \left(\{A,A\}\right) =0\,,
\end{equation}
because of the cyclicity of the trace. 
For later reference we state also the resulting formula  for triple products of classical symbols, 
\begin{align}
\label{eq:triple}
 A\#B\#C & \asymp  A_0B_0C_0 \;+\; \epsi A_1B_0C_0\;+\;\epsi A_0B_1C_0\;+\;\epsi A_0B_0C_1 \\ \nonumber
 & \quad-\,\tfrac{\imag\epsi}{2}\big( A_0\{B_0,C_0\}+\{A_0,B_0\} C_0+ \{A_0|B_0|C_0\}\big)+\mathcal{O}(\epsi^2)\,.
\end{align}
Here and in the following we use the shorthand
\begin{equation}\label{notation3}
 \{A_0|B_0|C_0\}  \;:=\; \partial_p A_0\cdot B_0\,\partial_q C_0  - \partial_q A_0\cdot B_0 \,\partial_p C_0\,.
\end{equation}
If $A = a \mathbf{1}_{\hfast}$ is a scalar multiple of the identity, then $A$ and all its derivatives commute with any $B$. As a consequence one can show that in this case
\begin{equation}\label{weylcommu}
A_0\# B_0 - B_0\# A_0 \asymp -\imag\epsi  \{ A_0, B_0\} \,+\, \Or(\epsi^3)\,.
\end{equation}
The fact that the remainder term in (\ref{weylcommu}) is of order $\epsi^3$ and not only $\epsi^2$ is at the basis of our higher order semiclassical approximations. It distinguishes  the Weyl quantization from other quantization rules.

We will also be interested in taking traces of pseudo-differential operators and of their symbols.
Denoting by $\mathcal{J}_1(\hfast)\subset \mathcal{B}(\hfast)$ the Banach space of trace-class operators on $\hfast$ with the trace-norm $\|A\|_1 := \tr_{\hfast}|A|$, the Moyal product restricts
to  maps
\[
\#: S^{k_1}(\epsi, \mathcal{J}_1(\hfast)) \times S^{k_2}(\epsi, \mathcal{B}(\hfast)) \to S^{k_1+k_2}(\epsi, \mathcal{J}_1(\hfast))
\]
and 
\[
\#:  S^{k_1}(\epsi, \mathcal{B}(\hfast)) \times S^{k_2}(\epsi, \mathcal{J}_1(\hfast)) \to S^{k_1+k_2}(\epsi, \mathcal{J}_1(\hfast))\,.
\]
The continuity of the product and of the asymptotic expansion  holds as well with respect to the trace-norm $\|\cdot\|_1$ on $\mathcal{J}_1(\hfast)$.

One can compute the trace of a product of Weyl operators by integrating the trace of the point-wise product of the symbols:
For scalar symbols $b\in S^0( \C)$ and  $a \in W^{\infty, 1}(\R^{2n} )= \big\{f\in C^\infty(\R^{2n})| \partial^\alpha f \in L^1(\R^{2n}) \mbox{ for all } \alpha\in\N_0^{2n}\big\}$    it is known, e.g.\ \cite{MR744068}, that $\hat a$ is trace class and that
\[
\tr_{L^2(\R^n)} \big(\hat a\hat b\big) \;=\;\frac{1}{(2\pi\epsi)^n} \int\D q\D p\;a(q,p)\,b(q,p)  \,.
\]
Now let $A\in W^{\infty, 1}(\R^{2n} ,\mathcal{B}(\hfast))$    and $B\in S^0(\mathcal{B}(\hfast))$
such that $\tr_{\hfast}(AB)\in L^1(\R^{2n})$.
Let $(\psi_j)_{j\in\N}$ and $(\varphi_k)_{k\in\N}$ be orthonormal bases of $L^2(\R^n)$ and $\hfast$ respectively. Then $a_{kl} := \langle \varphi_k, A\varphi_l\rangle_{\hfast}  \in W^{\infty, 1}(\R^{2n} )$ and 
and $b_{lk}:= \langle \varphi_l, B\varphi_k\rangle_{\hfast} \in S^0(\C)$ for all $k,l\in\N$.
Thus
 \begin{align}
\tr_\Hi \big(\hat A\hat B\big)  &=  \sum_{j,k}\langle \psi_j\otimes \varphi_k,\,\hat A\, \hat B\,\psi_j\otimes \varphi_k\rangle_\Hi
\;=\; \sum_{j,k,l}\langle \psi_j , \hat a_{kl} \hat b_{lk} \psi_j \rangle_{L^2(\R^n)}\nonumber \\
 &= \sum_{k,l} \frac{1}{(2\pi\epsi)^n} \int\D q\D p\;a_{kl}(q,p)b_{lk}(q,p)\nonumber \\
 &= \frac{1}{(2\pi\epsi)^n} \int\D q\D p\;\tr_{\hfast}(A(q,p) B(q,p))\;<\;\infty\,,\label{traceform}
\end{align}
where one should read this computation backwards to see that 
all expressions are well defined and that the equalities hold.


 \subsection{The almost invariant subspace}

With an isolated eigenvalue $e_0$ there is associated a subspace $\Pi^\epsi\Hi$ of the state space that is almost invariant under the action of $\hat H$.

\begin{prop}[\bf Super-adiabatic projection]
\label{satz:sat}
 Let Assumptions~1 and 2 hold. Then $e_0\in S^0(\R)$ and $\pi_0\in S^0(\mathcal{J}_1(\hfast))$. For $\epsi$ small enough, there exists an   orthogonal projection $\Pi^\epsi \in \mathcal{B}(\Hi)$ such that
   \begin{equation}
   \label{eq:bcommu}
  \big\| \big[\hat{H},\Pi^\epsi\big]\big\|  \;=\;\mathcal{O} (\epsi^\infty) \,.
 \end{equation}
 It is related to the band $e_0$ through
 \begin{equation}\label{pidiff}
 \| \Pi^\epsi - \hat \pi \|  = \Or(\epsi^\infty) \,,
 \end{equation}
  where $\pi\in S^0(\epsi,\mathcal{J}_1(\hfast))$ is a classical symbol with principal symbol $\pi_0$ taking values in the self-adjoint trace-class operators on $\hfast$.      $\Pi^\epsi$ is the spectral projection of $\hat \pi$ for its spectrum near one, i.e.\ $\Pi^\epsi = \chi_{[\frac12, \frac32]}(\hat \pi)$.  
The relevant diagonal block of the subprincipal symbol $\pi_1$ of $\pi$ is given by 
\begin{equation}
\label{eq:pidi}
\pi_0 \pi_{1}\pi_0 =\tfrac{\imag}{2}\,\pi_0 \{\pi_0,\pi_0\} \pi_0\,.
\end{equation}

\end{prop}
It is quite remarkable that for the following analysis we only need to know that such a projection $\Pi^\epsi=\hat\pi_0 +\Or(\epsi)$ exists and that the diagonal block of its subprincipal symbol is given by (\ref{eq:pidi}).

With $\pi\in S^0(\epsi,\mathcal{B}(\mathcal{H}_{\mathrm{f}}))$
the statement of Proposition~\ref{satz:sat}   was proven several times in the literature under slightly varying technical assumptions.  The strategy of proof is due  Nenciu and Sordoni \cite{NS:2004} based on work of   Helffer and Sj\"ostrand \cite{HS90}. 
The exact statement of Proposition~\ref{satz:sat} follows, for example, from inspecting Theorem~3.2 and its proof in
  \cite{Teufel:2003} using  the statements in Section~\ref{WeylForm} concerning 
 Moyal expansion in  $S^k(\epsi,\mathcal{J}_1(\hfast))$ and, as a starting point,   $\pi_0\in S^0(\mathcal{J}_1(\hfast))$. To see the latter, note that
   $\pi_0$ takes values in the rank one projections. Thus by induction starting with  
   \begin{equation}\label{pi0off}
\partial_j\pi_0 = \pi_0 ( \partial_j\pi_0)  \pi_0^\perp + \pi_0^\perp (\partial_j\pi_0)\pi_0\,
\end{equation}
all derivatives $\partial_z^\alpha\pi_0$ are of finite rank $\leq N_\alpha$ and therefore trace class. 
Hence  the fact that 
$\pi\in S^0(\epsi,\mathcal{B}(\mathcal{H}_{\mathrm{f}}))$ implies by
\[
\sup_{z\in\R^{2n}}\|\partial_z^\alpha\pi_0(z)\|_1 = \sup_{z\in\R^{2n}} \tr_{\hfast} |\partial_z^\alpha\pi_0(z)|\leq N_\alpha\sup_{z\in\R^{2n}} \|\partial_z^\alpha\pi_0(z)\| <\infty
\]  
for all $\alpha\in\N_0^{2n}$
 also  $\pi_0\in S^0(  \mathcal{J}_1(\hfast))$. 
  

\begin{coro}
\label{koro:time}
 Under the hypotheses of Proposition~\ref{satz:sat}
 the diagonal Hamiltonian $\Pi^\epsi\hat H\Pi^\epsi +\Pi^{\epsi\perp}\hat H\Pi^{\epsi\perp} $ is self-adjoint on $D(\hat H)$,
 the projection $\Pi^\epsi$ almost commutes with the unitary time evolution operator in the sense that
 \begin{equation} \label{Ucommu}
 \left\|  \left[\exp{- {\imag} \hat{H}s},\Pi^\epsi \right]\right\| \;=\;\mathcal{O} (\epsi^\infty \lvert s\rvert) 
 \end{equation}
and
 \begin{equation}\label{Hdiag}
 \left\|    \exp{- {\imag} \hat{H}s} -  \exp{- {\imag}( \Pi^\epsi \hat{H}\Pi^\epsi +\Pi^{\epsi\perp}\hat H\Pi^{\epsi\perp})s} \right\| \;=\;\mathcal{O} (\epsi^\infty \lvert s\rvert) \,.
 \end{equation}
 Moreover, for any $A\in S^k(  \mathcal{B}(\hfast))$   with $k>2n+1$ the operators $\hat A \,\hat\pi$ and $\hat A \,\Pi^\epsi$ are trace class with 
 \begin{equation}\label{traceest}
 \tr_\Hi (\hat A \,\hat\pi ) =  \Or(\epsi^{-n} \|A\|_{L^1} ) \quad\mbox{ and } \quad  \tr_\Hi ( \hat A \,\Pi^\epsi  )=  \Or(\epsi^{-n} \|A\|_{L^1})\,.
\end{equation}
\end{coro}
\begin{proof} 
The self-adjointness of the projected Hamiltonian on $D(\hat H)$ follows from the fact that $[\hat H,\Pi^\epsi]$ is a bounded operator.
A standard Duhamel expansion yields (\ref{Ucommu}) and (\ref{Hdiag}).
 For the last statement note that by (\ref{traceform})  the estimate 
 \[
\int_{\R^{2n}} \hspace{-3pt}\D q\D p \;|\tr_{\hfast}(A(q,p)\, \pi(q,p) )|  \leq 
\int_{\R^{2n}}  \hspace{-3pt}\D q\D p \;|\tr_{\hfast}( \pi(q,p))|\,\|A(q,p)\| \leq C\,\|A\|_{L^1}  < \infty 
\]
implies $\hat A \,\hat\pi\in \mathcal{J}_1(\Hi)$. 
 Since $\Pi^\epsi$ is the spectral projection of $\hat \pi$ to its spectrum near one, there is a bounded continuous function $g:\R\to \R$ with $\|g(\hat \pi) - {\bf 1}_\Hi\| = \Or(\epsi^\infty)$ such that $\Pi^\epsi = g(\hat \pi)\hat \pi$. This implies that also $\hat A\,\Pi^\epsi $ is trace class and that 
 \[
 \tr_\Hi(\hat A \, \Pi^\epsi  )=  \tr_\Hi(\hat A \,g(\hat \pi)\hat \pi) = \tr_\Hi(\hat A \, \hat \pi)(1 + \Or(\epsi^\infty))\,.
\] 
\end{proof}

Note that (\ref{Hdiag}) is the starting point for an approximation based on the concept of effective Hamiltonians. The idea is to map the range of $\Pi^\epsi$  
unitarily to $L^2(\R^n)$ using an eigenfunction $\varphi_0(q,p)$ in the range of $\pi_0(q,p)$. 
The image of   the block $\Pi^\epsi\hat H\Pi^\epsi$ of $\hat H$ is then 
a pseudo-differential operator $\hat h_{\rm eff}$, the effective Hamiltonian.  Its scalar symbol $h_{\rm eff}$ 
has an asymptotic expansion that can be, in principle, computed to any order. But it depends on a gauge, namely the choice of $\varphi_0(q,p)$. For details on this approach we refer to \cite{HS90,Teufel:2003}.

\section{The classical Hamiltonian system}\label{sec-clas}

For the semiclassical analysis within the subspace $\Pi^\epsi\Hi$ we can now focus on the restricted Hamiltonian $\Pi^\epsi \hat{H}\Pi^\epsi$. 
The central object is the following  scalar Hamiltonian function
\begin{equation}\label{hscdef}
\hsc(z)  :=e_0(z) + \xi\cdot z + \epsi \,\tr_{\Hi_{\rm f}}(H_1(z)\pi_0(z))+\epsi M(z),
\end{equation}
where 
\begin{equation} \label{Mdef}
M:=\tfrac{\imag}{2}\,\tr_{\Hi_{\rm f}}\left(\{\pi_0|H_0|\pi_0\}\right)  \,.
\end{equation}
We will see that $M$ can be written equivalently as
\begin{align}\label{Malt}
M \;&=  \;\tfrac{\imag}{2}\,\tr_{\Hi_{\rm f}}\left(\{\pi_0|H_0-e_0|\pi_0\}\right) \;=\;  \tfrac{\imag}{2}\,\tr_{\Hi_{\rm f}}\left(\{\pi_0|H_0-e_0|\pi_0\}\pi_0\right)\nonumber\\
&= \;-\tfrac{\imag}{2}\,\tr_{\Hi_{\rm f}}\left(\pi_0\{\pi_0,H_0-e_0\}\right)   \,,
\end{align}
which is the expression found in \cite{MR1376438}, who explain how to relate it also to \cite{MR1133847}.
The scalar Hamiltonian $h$ has the property that its quantization approximates the action of the full operator $\hat H$ on the range of the projection $\Pi^\epsi$ up to terms of order $\epsi^2$,
\[
\Pi^\epsi \hat{H}\Pi^\epsi  = \Pi^\epsi \,\hat \hsc \,\Pi^\epsi+\Or (\epsi^2)\,.
\]

\begin{prop} The symbol $\tilde \hsc(z) := \hsc(z) - \xi\cdot z$ belongs to $S^0(\epsi,\C)$ 
and $h$ satisfies
\begin{equation}\label{hscHapprox}
\pi \# \hsc  \# \pi  - \pi \# H \# \pi = \Or(\epsi^2) \quad\mbox{ in } \,S^0(\epsi, \mathcal{B}(\Hi_{\rm f}))\,.
\end{equation}
As a consequence,
\begin{equation}\label{eq:hsc}
\left\|  \Pi^\epsi \,\hat \hsc \,\Pi^\epsi - \Pi^\epsi \,\hat H \,\Pi^\epsi\right\|  = \Or(\epsi^2)\,.
\end{equation}
\end{prop}
Note that, in general, there is no scalar symbol $\hsc $ such that (\ref{hscHapprox}) holds with an error $\Or(\epsi^3)$. Neither does an analogous scalar symbol $\hsc $ exist in the case of a degenerate eigenvalue $e_0$. Also note that in the introduction we did not distinguish between $H$ and its principal symbol~$H_0$. Indeed, if we replace $H_0$ in all the following by $H$ and ignore terms containing $H_1$, all the results hold equally. Of course then $e_0$ and $\pi_0$ need to be eigenvalues and spectral projections of $H$. In concrete applications, however,  it is typically easier to split off the principal symbol.
 
 \begin{proof} The claim that $\tilde \hsc \in S^0(\epsi,\C)$ is an immediate consequence of the assumptions.
Since the linear part $\xi\cdot z$ appears also as a scalar in $H$, it suffices to check 
(\ref{hscHapprox}) for $\tilde h$ and $\tilde H = H_0 + \epsi H_1$.
According to (\ref{eq:triple}) the  asymptotic expansion of ${\pi}\#{\tilde H}\#{\pi}$ starts with  
\begin{align*}
\nonumber
{\pi}\#{\tilde H}\#{\pi}\;&= \; e_0\pi_0+\epsi\big( \pi_0 H_1\pi_0+e_0 \pi_1 \pi_0+e_0 \pi_0\pi_1  \\& \quad
\nonumber - \tfrac{\imag}{2}\pi_0\{H_0,\pi_0\}-  \tfrac{\imag}{2}\{\pi_0,H_0\}\pi_0
-\tfrac{\imag}{2}\{\pi_0|H_0|\pi_0\} \big) +\mathcal {O}(\epsi^2)\\ \nonumber
&=\; e_0\pi_0+\epsi \big( \tr_{\hfast}\left(\pi_0 H_1\right) \pi_0+e_0 \pi_1 \pi_0+e_0 \pi_0\pi_1  \\&  \quad   - \tfrac{\imag}{2}\pi_0\{H_0,\pi_0\}-  \tfrac{\imag}{2}\{\pi_0,H_0\}\pi_0
-\tfrac{\imag}{2}\{\pi_0|H_0|\pi_0\}\big) +\mathcal {O}(\epsi^2)\,.
\end{align*}
On the other hand 
\begin{align*}
 \nonumber
\pi \# \tilde \hsc  \# \pi \;&= \; e_0\pi_0 +\epsi \big( \tr_{\hfast}\left(\pi_0 H_1\right) \pi_0+e_0 \pi_1 \pi_0+e_0 \pi_0\pi_1
\\ \nonumber  & \quad   +\; M\,\pi_0   -\tfrac{\imag}{2}\pi_0 \{e_0,\pi_0\}-  \tfrac{\imag}{2}\{\pi_0,e_0\}\pi_0
-\tfrac{\imag}{2}\{\pi_0|e_0|\pi_0\}  \big)+\mathcal {O}(\epsi^2)\, .
\end{align*}
It remains to show that
\begin{equation}\label{Mform}
\tfrac{2}{\imag}M\pi_0\;=\; - \pi_0\{H_0-e_0,\pi_0\}-\{\pi_0,H_0-e_0\}\pi_0
-\{\pi_0|H_0-e_0|\pi_0\}\,.
\end{equation}
To see this note that (\ref{cyctrace}) and  (\ref{pi0off}) imply
\begin{align*}
M\pi_0\; &= \;\tfrac{\imag}{2} \tr_{\Hi_{\rm f}}\left(\{\pi_0|H_0|\pi_0\}\right)\pi_0\;
\stackrel{(\ref{cyctrace})}{=}\; \tfrac{\imag}{2} \tr_{\Hi_{\rm f}}\left(\{\pi_0|H_0-e_0|\pi_0\}\right)\pi_0\\
&\stackrel{(\ref{pi0off})}{=}\; \tfrac{\imag}{2} \tr_{\Hi_{\rm f}}\left(\pi_0 \{\pi_0|H_0-e_0|\pi_0\}\pi_0\right)\pi_0
\;=\; \tfrac{\imag}{2} \pi_0 \{\pi_0|H_0-e_0|\pi_0\} \pi_0\\
&\stackrel{(\ref{pi0off})}{=}\;  \tfrac{\imag}{2}   \{\pi_0|H_0-e_0|\pi_0\} \,.
\end{align*}
Moreover
\[
0 =   \partial_j \big( (H_0-e_0) \pi_0\big)\pi_0 =  \partial_j  (H_0-e_0)  \pi_0 +   (H_0-e_0) \partial_j  \pi_0\pi_0
\]
implies that 
\[
 \{\pi_0,H_0-e_0\}\pi_0 =- \{\pi_0| H_0-e_0|\pi_0\}\pi_0 =- \{\pi_0|H_0-e_0|\pi_0\}\,,
\]
which shows also  (\ref{Malt}).
By the same reasoning we have also 
$ \pi_0\{H_0-e_0,\pi_0\}=- \{\pi_0| H_0-e_0|\pi_0\}$
and thus (\ref{Mform}) follows.
 \end{proof}
 
In order to obtain semiclassical approximations up to errors of order $\epsi^2$, we need to take into account that 
the restriction to the range of $\Pi^\epsi$ also induces a modified symplectic form $\omega_\epsi $ on $\R^{2n}$ given by 
\begin{equation}\label{omegadef}
{\omega_\epsi} \;:=\;  \omega_0 + \epsi \,\Omega \;=:\;
 \left( \begin{array}{cc} 0 &  E_n\\  -E_n & 0\end{array}\right)\;+\;\epsi \left( \begin{array}{cc} \Omega^{qq} &  \Omega^{qp}\\  \Omega^{pq} & \Omega^{pp}\end{array}\right)
\,,
\end{equation}
where the
  components of $\Omega$ in the canonical basis are,
  with  $z=(q,p)$ and $\alpha,\beta =1,\ldots,2n$,
\[
  \Omega_{\alpha\beta} \;:=\; -\,\imag \,\tr_{\hfast}\left(\pi_0[\partial_{z_\alpha}\pi_0,\partial_{z_\beta}\pi_0] \right) \,.
\]
As is well known and will be shown also in  Proposition~\ref{berryprop}, $\Omega$ is the curvature 2-form of the Berry connection.

The  Liouville measure $\lambda_\epsi $ associated with the symplectic form $\omega_\epsi $ has the expansion 
\begin{equation}\label{liouville}
\lambda_\epsi=  \left( 1+\imag \epsi \,\tr_{\hfast}\left(\pi_0\{\pi_0,\pi_0\}\right)+\mathcal{O}(\epsi^2)\right) \df q^1\wedge\dots\wedge\df p^n. 
\end{equation}
We postpone a discussion of these objects to Section~\ref{section:correm}. There we  prove that $\omega_\epsi$ is indeed a symplectic form and that its Liouville measure is given by (\ref{liouville}).
First, however, we present and prove our main results concerning semiclassical approximations.

\section{Semiclassical approximations}
\label{sec:sc}

We now explain how to approximate quantum mechanical expectation values 
\[
\tr_\Hi \left( \rho \,\hat a\right)
\]
 of semiclassical observables 
$\hat a := \hat{ a}  \otimes \mathbf{1}_{\hfast}$ for the slow degrees of freedom in terms of the classical Hamiltonian system described in the preceding section. Of course this can only work for states 
$\rho$ that belong to the range of $\Pi^\epsi$.

\subsection{Expectation values for stationary states}

In this section we consider stationary states that are suitable functions of the Hamiltonian,~i.e.\
\[
\rho = f( \hat H)\,.
\]
Since we make use of the Helffer-Sj\"ostrand formula for analyzing functions of self-adjoint operators, a natural 
space of admissible functions $f$ is the space
\[
\mathcal{A} := \{ f\in C^\infty(\R)\;|\; \exists \beta < 0\,:\; \sup_{x\in\R} |\langle x\rangle^{n-\beta  } \,f^{(n)}(x)| <\infty\;
\mbox{ for all } n\in \N_0\, \}\,.
\] 
The following theorem shows that we can approximate the quantum mechanical expectation values
by a classical phase space average up to errors of order $\epsi^2$ if we use the correct classical Hamiltonian $\hsc $ and the correct Liouville measure~$\lambda_\epsi $.
An important application of this result is the possibility  to express quantum equilibrium distributions in terms of the classical system. 


\begin{theo}[\bf Equilibrium distributions]
\label{satz:exp}
Let Assumptions~1 and 2 of Section~\ref{sec:sett} hold. Then for  $f\in \mathcal{A}$ and $a\in S^k(\C) $ with $k>2n+1$ it holds that
\[
\tr_\Hi \left(\Pi^\epsi f(\hat H)\, \hat a\right) \;=\; \frac{1}{(2\pi\epsi)^n}\left( \int \D \lambda_\epsi \; f(\hsc (q,p)) \, a(q,p) \;+\;\Or\left(\epsi^{2}\|a\|_{L^1}\right)\right)\,.
\]
\end{theo}
Note that the precise error estimate in terms of $\|a\|_{L^1}$ is very useful when taking e.g.\ a thermodynamic limit, i.e.\ when looking at a sequence of observables $a_n$ that extend over larger and larger regions of phase space. This is done explicitly in Theorem~\ref{pressure}, where we use Theorem~\ref{satz:exp} to compute the free energy per unit volume in the Hofstadter model. In that model it is also explicit, that only the phase space volume related to each band is modified, while the total phase space volume is unchanged as the total integrated curvature of all bands is zero. 
\begin{proof} 
According to Corollary~\ref{koro:time}, $\Pi^\epsi\hat a$ is trace class with $\tr_\Hi (\Pi^\epsi\hat a) = \Or(\epsi^{-n}\|a\|_{L^1})$
and thus we can estimate expressions of the form $\tr_\Hi (R \,\hat a\,\Pi^\epsi)$ with $R\in \mathcal{L}(\Hi)$ by  $C\epsi^{-n}\|R\|\,\|a\|_{L^1}$ with a constant $C$ independent of $\epsi$ and $R$. Recall that $\| [  \hat H ,\,\Pi^\epsi] \|  = \Or(\epsi^\infty)$ and, since    $\hsc $ is scalar,   $\| [  \hat \hsc,\,\Pi^\epsi] \| =\| [  \hat \hsc,\,\hat \pi] \| +\Or(\epsi^\infty)  = \Or(\epsi)$.
In the following computation we use these observations together with the statements of Lemma~\ref{lem:est} and Lemma~\ref{lem:est2} from the appendix where indicated,
\begin{align*}
\tr_\Hi \left( f(\hat H)\, \hat a\,\Pi^\epsi\right)   \;
&\stackrel{(\ref{esti1})}{=}  \; \tr_\Hi \left(\Pi^\epsi  f(\hat H )\, \Pi^\epsi\, \hat a \right)\;+\;\Or(\epsi^\infty\|a\|_{L^1})\\
&\stackrel{(\ref{esti2})}{=}  \; \tr_\Hi \left(\Pi^\epsi  f(\Pi^\epsi \hat H \Pi^\epsi )\, \Pi^\epsi \,\hat a \right) \;+\;\Or(\epsi^\infty\|a\|_{L^1})\\
&\stackrel{(\ref{esti3})}{=}  \; \tr_\Hi \left(\Pi^\epsi  f(\Pi^\epsi \hat \hsc \Pi^\epsi )\, \Pi^\epsi \,\hat a \right) \;+\;\Or(\epsi^{2-n}\|a\|_{L^1})\\
&\stackrel{(\ref{esti2})}{=}  \; \tr_\Hi \left(\Pi^\epsi  f( \hat \hsc  )\, \Pi^\epsi \,\hat a \right) \;+\;\Or(\epsi^{2-n}\|a\|_{L^1})\\
& \;=  \; \;\tr_\Hi \left(\hat\pi\, f( \hat \hsc   )\, \hat\pi\, \hat a  \right) \;+\;\Or(\epsi^{2-n}\|a\|_{L^1})\,.
\end{align*}
Next note that for scalar symbols the functional calculus for pseudo-differen\-tial operators implies 
that $\| f( \hat \hsc  ) - \widehat{f(   \hsc   )}\|= \Or(\epsi^2)$, c.f.\ e.g.\ Chapter~8 in \cite{MR1735654}. Since this is an important and nontrivial step, we sketch a proof of this statement in Lemma~\ref{lem-func} in the appendix.
Hence with $f_h := f\circ h \in S^0(\epsi,\R)$ we have
\begin{align}\label{form1}
\tr_\Hi \left(\Pi^\epsi f(\hat H )\, \hat a\right)   
&\;=  \;    \;\tr_\Hi \left( \hat\pi \widehat{f(   \hsc  )} \hat\pi\hat a\right) \;+\;\Or(\epsi^{2-n}\|a\|_{L^1})\\
&\stackrel{(\ref{traceform})}{=}  \; \frac{1}{(2\pi\epsi)^n} \int \D q\D p\; \;\tr_{\Hi_{\rm f}}\left(  ( \pi\#  f_h\#\pi) (q,p)\,  a  (q,p)\right)\;+\;\Or(\epsi^{2-n}\|a\|_{L^1})\nonumber\\
& \; =  \;  \frac{1}{(2\pi\epsi)^n}\int \D q\D p\;  a(q,p)\;\tr_{\Hi_{\rm f}}\left(  (  \pi\# f_h\#\pi ) (q,p)\right)\;+\;\Or(\epsi^{2-n}\|a\|_{L^1})
\,,\nonumber
\end{align}
where  $ \pi\# f_h\#\pi \in S^0(\epsi, \mathcal{J}_1(\hfast))$.  

It remains to compute $\tr_{\hfast}\left( \pi\# f_h\#\pi   \right)$ up to $\mathcal{O}(\epsi^2)$ in $S^0(\epsi, \mathcal{J}_1(\hfast))$. 
Using (\ref{eq:triple}) and the fact that $f_h$ is scalar, we have that  the expansion of 
 $F := \pi\# f_h\#\pi$ starts with
$F_0 = \pi_0 f_{h0}$ and 
\[
F_1\;=\; \pi_0 f_{h1}\pi_0+ f_{h0}\pi_1\pi_0+f_{h0}\pi_0\pi_1   - \tfrac{\imag}{2}\pi_0\{f_{h0},\pi_0\}-  \tfrac{\imag}{2}\{\pi_0, f_{h0}\}\pi_0
-\tfrac{\imag}{2}f_{h0}\{\pi_0,\pi_0\}\,.
\]
Taking the trace we get
\begin{align*}
\tr_{\hfast}\left( \pi\# f_h\#\pi   \right) \; &= \;  \tr_{\hfast}\left( \pi_0 f_{h0} +\epsi F_1   \right) \;+\;\Or(\epsi^2)
 \\
&= \;   f_{h0} +\epsi \left(  \tr_{\hfast}\left( \pi_0 F_1 \pi_0   \right)+\tr_{\hfast}\left( \pi_0^\perp F_1 \pi_0^\perp   \right)\right)\;+\;\Or(\epsi^2) \\
&=  \;  f_{h0}+\epsi f_{h1} + \epsi  \,2 f_{h0}  \tr_{\hfast}\left( \pi_0 \pi_1 \pi_0   \right) 
\;+\;\Or(\epsi^2)\,,
\end{align*}
where we used that
for scalar symbols    
 \[  
  \{f_h,\pi_0\} = \pi_0 \{f_h,\pi_0\}\pi_0^\perp +\pi_0^\perp \{f_h,\pi_0\} \pi_0 
\]
has vanishing trace 
 and   that also $\tr_{\hfast}\{\pi_0,\pi_0\}=0$, see (\ref{cyctrace}). With the expression (\ref{eq:pidi}) for $\pi_1$ the final result  is
\begin{equation*}
\tr_{\hfast}\left( \pi\# f_h\#\pi   \right)= f_h \left(1 + \imag \epsi \tr_{\hfast}\left( \pi_0 \{ \pi_0,\pi_0\}    \right)  \right)\;+\;\Or(\epsi^2)
\,.
\end{equation*}
Inserting this into (\ref{form1}) and comparing with (\ref{liouville}) completes the proof.
\end{proof}

\subsection{Egorov theorem and transport of Wigner functions}

For non-stationary states $\rho(t) =  \ez^{-\imag \hat H \frac{t}{\epsi}} \,\rho \, \ez^{\imag \hat H \frac{t}{\epsi}}$ starting in the range of $\Pi^\epsi$, i.e.\ $\rho =\Pi^\epsi\rho\Pi^\epsi$, (\ref{Ucommu}) implies that
\begin{align*}
\Pi^\epsi \,\rho(t)\,\Pi^\epsi  \; &= \;   \Pi^\epsi \,\ez^{-\imag \hat H \frac{t}{\epsi}} \,\rho \, \ez^{\imag \hat H \frac{t}{\epsi}}\,\Pi^\epsi
= \ez^{-\imag \hat H \frac{t}{\epsi}} \,\Pi^\epsi \rho\Pi^\epsi \, \ez^{\imag \hat H \frac{t}{\epsi}} +\Or(|t|\epsi^\infty)
\\&= \; \rho(t)+\Or(|t|\epsi^\infty)\,.
\end{align*}
For semiclassical observables $\hat a = \hat a\otimes \mathbf{1}_{\hfast}$   we thus obtain
\[
\tr_\Hi ( \rho(t) \,\hat a )\;=\;\tr_\Hi ( \rho \;\Pi^\epsi \underbrace{ \ez^{\imag \hat H \frac{t}{\epsi}} \,\hat a \, \ez^{-\imag \hat H \frac{t}{\epsi}}}_{=:A(t)} \Pi^\epsi )+\Or(|t|\epsi^\infty)\,.
\]
The Egorov theorem shows that on the range of $\Pi^\epsi$ one can approximate the quantum mechanical evolution $A(t)$ 
by evolving the symbol $a$ along the classical flow $\phi^t_\epsi$ up to errors of order $\epsi^2$.
Its proof for scalar Hamiltonians follows from the relation (\ref{weylcommu}) between commutators of operators and the Poisson bracket for the symbols.
As the following proposition shows, a similar statement holds in our case, but with a modified Poisson bracket.

\begin{prop}\label{commprop}
Let Assumptions~1 and 2 of Section~\ref{sec:sett} hold. Then for $\epsi$ small enough, 
the Hamiltonian vector field   $X_h^\alpha = -(\omega_\epsi)^{\alpha\beta} \partial_\beta h$ is a classical symbol in $S^0(\epsi,\R^{2n})$. There is a $r\in\N$ such that for any  $a \in S^0(\epsi,\mathbb{C})$ it holds that
\begin{equation}\label{CommuPoisson}
\left\|  \Pi^\epsi \left( \tfrac{\I}{\epsi}\big[ \hat H, \hat a \big]   \;-\;   {\rm Op}^{\rm W} ( X_h\cdot \nabla a
 ) \right)  \Pi^\epsi \right\| \;=\; \Or(\epsi^2\|a\|_{0,r})
 \,.
\end{equation}
Here $X_h\cdot \nabla a = \{ h,a\}_{\omega_\epsi}$ is just the Poisson bracket with respect to the modified symplectic form. For $a(z) = \eta\cdot z$ with $\eta\in\R^{2n}$ formula (\ref{CommuPoisson}) holds with an error of order $\Or(\epsi^2|\eta|)$.

\end{prop}
\begin{proof}   The statement about $X_h$ follows immediately  from $\Omega_{\alpha\beta}\in S^0(\R)$ 
and $\nabla h\in S^0(\epsi, \R^{2n})$. Let  $a \in S^0(\epsi,\mathbb{C})$, then with (\ref{eq:bcommu}) we have 
\[
\tfrac{\I}{\epsi} \Pi^\epsi \big[ \hat H, \hat a \big] \Pi^\epsi =  \tfrac{\I}{\epsi}\big[  \Pi^\epsi \hat H  \Pi^\epsi ,  \Pi^\epsi \hat a  \Pi^\epsi\big]  \;+\;\Or(\epsi^\infty\|\hat a\| )\,.
\]
In order to replace $\hat H$ by $\hat h$, observe that 
\[
 \tfrac{\I}{\epsi} \Pi^\epsi (\hat{H}- \hat \hsc ) \Pi^\epsi = \tfrac{\I}{\epsi} {\rm Op}^{\rm W}( \pi \# H \# \pi  - \pi \#  \hsc \# \pi )\;+\;\Or(\epsi^\infty)
\]
and that by (\ref{hscHapprox})
\[
h_2 := \epsi^{-2}( \pi \# H \# \pi  - \pi \#  \hsc  \# \pi )\in S^0(\epsi,\mathcal{B}(\hfast))
\]
 satisfies $\pi \# h_2 \# \pi = h_2 + \Or(\epsi^\infty)$. Hence  its principal symbol 
satisfies
\[
(h_2)_0 = \pi_0 \,(h_2)_0\,\pi_0 
\]
and we obtain for sufficiently large $r\in\N$ with $r\geq 2n+1$ that
\begin{align*}
 \tfrac{\imag}{\epsi} \left[ \Pi^\epsi (\hat{H}- \hat \hsc ) \Pi^\epsi, \Pi^\epsi\hat{ a } \Pi^\epsi \right]  \; 
&=  \; \imag\epsi  \left[\hat h_2,  \hat \pi \,\hat a\,\hat\pi \right]  +\Or(\epsi^\infty\|\hat a\|)\\
&=  \; \imag\epsi \, {\rm Op}^{\rm W} \left(\left[(h_2)_0,  \pi_0  a   \pi_0 \right] \right)  +\Or(\epsi^2\| a\|_{0,r}) = \Or(\epsi^2\|a\|_{0,r})\,.
\end{align*}
The remaining term can be rearranged in the following way,
\begin{eqnarray*} \lefteqn{\hspace{-5mm}
 \tfrac{\imag}{\epsi}  \left[ \Pi^\epsi \hat\hsc  \Pi^\epsi, \Pi^\epsi\hat{ a } \Pi^\epsi \right]  
\;=\; \tfrac{\imag}{\epsi}  \Pi^\epsi [ \hat\hsc  ,\hat{ a }  ]  \Pi^\epsi  -
\tfrac{\imag}{\epsi}  \Pi^\epsi \left(\hat\hsc  \Pi^{\epsi\perp} \hat{ a } -   \hat{ a } \Pi^{\epsi\perp}  \hat\hsc \right) 
\Pi^\epsi} \\
&=& \tfrac{\imag}{\epsi}  \Pi^\epsi  [ \hat\hsc  ,\hat{ a }  ]  \Pi^\epsi  +
\tfrac{\imag}{\epsi}  \Pi^\epsi \left( [\hat\hsc , \Pi^{\epsi\perp}]\,[ \hat{ a } ,\Pi^{\epsi\perp}]- [\hat{ a }, \Pi^{\epsi\perp}]\,[  \hat\hsc ,\Pi^{\epsi\perp}] \right) \Pi^\epsi \\
&=& \tfrac{\imag}{\epsi}  \Pi^\epsi  [ \hat\hsc  ,\hat{ a }  ]  \Pi^\epsi  -\imag \epsi\,
 \Pi^\epsi \left[\tfrac{\imag}{\epsi}  [\hat\hsc , \Pi^\epsi]\,,\,\tfrac{\imag}{\epsi} [ \hat{ a } ,\Pi^\epsi] \right] \Pi^\epsi \,.
\end{eqnarray*}
Since $\hsc $ and $ a $ are scalar, (\ref{weylcommu}) implies
\begin{align*}
\tfrac{\imag}{\epsi} [ \hat\hsc  ,\hat{ a }  ]   \; &= \;  {\rm Op}^{\rm W} ( \{\hsc  ,{ a }  \}) +\Or(\epsi^2\| a\|_{0,r})\\
\tfrac{\imag}{\epsi} [ \hat\hsc  ,\Pi^\epsi  ]   \; &= \;  {\rm Op}^{\rm W} ( \{\hsc  ,\pi_0  \})+\Or(\epsi )\\
\tfrac{\imag}{\epsi} [ \hat{ a }  ,\Pi^\epsi  ]  \;  &= \;  {\rm Op}^{\rm W} ( \{  a  ,\pi_0  \})+\Or(\epsi\| a\|_{0,r})\,.
\end{align*}
Here and in the following steps a priori   different $r$s might be necessary. We agree to denote by $r$ the largest one  
 that works in all estimates.
Summing up the estimates we got up to now we have
\[
\tfrac{\I}{\epsi} \Pi^\epsi \big[ \hat H, \hat a \big] \Pi^\epsi \;=\;
  \Pi^\epsi \; {\rm Op}^{\rm W} \left( \{\hsc  ,{ a }  \}  -\imag\epsi\, [  \{\hsc  ,\pi_0  \},  \{  a  ,\pi_0  \}]
  \right)\;
   \Pi^\epsi+ \Or(\epsi^2\| a\|_{0,r})\,.
\]
The subprincipal symbol on the right hand side can be replaced by a scalar one at leading order, because of the outside projections,  
\begin{eqnarray*}\lefteqn{\hspace{-1.5cm}
  \Pi^\epsi \; {\rm Op}^{\rm W} \left(  [  \{\hsc  ,\pi_0  \},  \{  a  ,\pi_0  \}]
  \right)\,\Pi^\epsi \;=\;
  \Pi^\epsi\,  {\rm Op}^{\rm W} \left(  \pi_0[  \{\hsc  ,\pi_0  \},  \{  a  ,\pi_0  \}]\pi_0
  \right)\,\Pi^\epsi +\Or(\epsi\| a\|_{0,r})}\\
  & = & \Pi^\epsi\,  {\rm Op}^{\rm W} \big(  \tr_{\hfast}\left( \pi_0[  \{\hsc  ,\pi_0  \},  \{  a  ,\pi_0  \}] \big) \pi_0
  \right)\,\Pi^\epsi +\Or(\epsi\| a\|_{0,r})
  \\
  & =& \Pi^\epsi\,  {\rm Op}^{\rm W}  \big(  \tr_{\hfast}  ( \pi_0[  \{\hsc  ,\pi_0  \},  \{  a  ,\pi_0  \}] )  
   \big)\,\Pi^\epsi +\Or(\epsi\| a\|_{0,r})
\end{eqnarray*}
and thus
\[
\tfrac{\I}{\epsi} \Pi^\epsi \big[ \hat H, \hat a \big] \Pi^\epsi \;=\;
  \Pi^\epsi \; {\rm Op}^{\rm W} \big( \{\hsc  ,{ a }  \}  -\imag\epsi\, \tr_{\hfast}  ( \pi_0[  \{\hsc  ,\pi_0  \},  \{  a  ,\pi_0  \}] )  
  \big)\;
   \Pi^\epsi+ \Or(\epsi^2\| a\|_{0,r})\,.
\]
Comparing this with
\begin{align*}
(\nabla a) \cdot X_h   &=\; (\partial_q  a ,\partial_p  a )
 \left( \begin{array} {cc} - \epsi \,\Omega^{pp} &         E_n+\epsi \Omega^{pq} \\    -E_n+\epsi \Omega^{qp} &- \epsi \,\Omega^{qq} \end{array} \right)  \left( \begin{array}{c} \partial_q \hsc \\ \partial_p \hsc     \end{array}   \right) \;+\;\Or(\epsi^2\| a\|_{0,1})\\[1mm]
 &= \;  \{\hsc  ,{ a }  \}  -\imag\epsi\,  \tr_{\hfast}\left( \pi_0[  \{\hsc  ,\pi_0  \},  \{  a  ,\pi_0  \}]\right)\;+\;\Or(\epsi^2\| a\|_{0,1})
\end{align*}
completes the proof for $a \in S^0(\epsi,\mathbb{C})$.
Now notice that the result implies that on the level of classical symbols we have
\begin{equation}\label{helleq}
\tfrac{\I}{\epsi}\pi \# \left( H  \#  a -   a \# H \right)  \#  \pi = \pi \#  (X_h \cdot \nabla a)  \#  \pi + \Or(\epsi^2)\,.
\end{equation}
However, as an algebraic relation it holds   for $a$ in any of the symbols classes $S^k(\R)$, in particular also for $a(z) = \eta\cdot z$. But then $\tfrac{\I}{\epsi}(H  \#  a -   a \# H) = \{ H, a\} \in S^0(\epsi,\mathcal{L}(\hfast))$ and $ (X_h \cdot \nabla a) = X_h\cdot \eta\in S^0(\epsi,\R)$ imply that also the remainder term is $\Or(\epsi^2)$ in $S^0(\epsi,\mathcal{L}(\hfast))$.

Note that checking  relation (\ref{helleq}) directly by expanding both sides in powers of $\epsi$  is an alternative way to prove the proposition. This was done in \cite{diss} and leads to much more painful computations than the argument given here. 
\end{proof}

\begin{theo}[\bf Egorov theorem]
\label{satz:ego}
Let Assumptions~1 and 2 of Section~\ref{sec:sett} hold. Then the Hamiltonian  flow $\phi^t_\epsi$  of $(\hsc , \omega_\epsi )$ exists globally and for   any $a \in S^0(\epsi,\mathbb{C})$ it holds that
\begin{equation}
\nonumber
\left\lVert \Pi^\epsi \left( A(t)-\widehat{a\circ \phi^t_\epsi} \right) \Pi^\epsi \right\rVert=\mathcal{O}(\epsi^2 )
\end{equation} 
 uniformly on bounded time intervals, where
 \[
A(t) := \ez^{\imag \hat H \frac{t}{\epsi}} \,\hat a \, \ez^{-\imag \hat H \frac{t}{\epsi}}\,.
 \]
\end{theo}

\begin{proof} Since the Hamiltonian vector field $X_h$ is smooth with all derivatives uniformly bounded, the corresponding flow  $(\phi^t_\epsi)_{t\in \R}$ exists globally. Each map $\phi^t_\epsi:\R^{2n}\to \R^{2n}$ is a $C^\infty$-diffeomorphism with all derivatives bounded. Hence $t\mapsto a(t) := a \circ  \phi^t_\epsi\in S^0(\epsi,\C)$ is smooth and uniformly bounded on bounded intervals in time.
In order to apply Proposition~\ref{commprop} we use the standard Duhamel argument and  (\ref{Ucommu}) to show that
 \begin{eqnarray*}\lefteqn{\hspace{-5mm}
 \Pi^\epsi \left(  {A} (t)-\widehat{ a(t)}\right)\Pi^\epsi
 = \Pi^\epsi \int_0^t \df s \frac{\df}{\df s}\left(   \ez^{\imag \hat H \frac{s}{\epsi}} \widehat{ a(t-s)}        \ez^{-\imag \hat H \frac{s}{\epsi}}\right)       \Pi^\epsi} \\ 
 &=& \Pi^\epsi \int_0^t \df s \,   \ez^{\imag \hat H \frac{s}{\epsi}}  \left( \tfrac{\imag}{\epsi}\left[ \hat{H},\widehat{ a(t-s)} \right] + \tfrac{\df}{\df s}  \widehat{ a(t-s)}  \right)  \ez^{-\imag \hat H \frac{s}{\epsi}}  \Pi^\epsi
 \\ 
& =& \int_0^t \df s \,   \ez^{\imag \hat H \frac{s}{\epsi}}  \left( \tfrac{\imag}{\epsi}  \Pi^\epsi  \left[\hat{H} ,\widehat{ a(t-s)}  \right] \Pi^\epsi+ \tfrac{\df}{\df s}  \Pi^\epsi \widehat{ a(t-s)}  \Pi^\epsi \right)  \ez^{-\imag \hat H \frac{s}{\epsi}} \\&& +\; \Or\big(\epsi^\infty  t^2  \sup_{s\in[0,t]}\|\widehat{ a(s)} \| \big)\,.
\end{eqnarray*}
Hence the claim follows from (\ref{CommuPoisson}) and the fact that $a(t) := a \circ  \phi^t_\epsi $ satisfies $\frac{\D}{\D t} a(t) = X_h\cdot \nabla a(t)$.
\end{proof}

As a corollary we get the following result on the transport of Wigner functions, which is dual to the Egorov theorem. 
For $\psi\in  L^2(\mathbb{R}^n,\hfast)$ the Wigner transform $W^\psi$ of $\psi$ is defined by
  \begin{equation*}
  W^\psi(q,p) := \frac{1}{(2\pi)^n}\int_{\mathbb{R}^n}\df x\, \exp{\imag p\cdot x}    \psi\left(q-\tfrac{\epsi}{2}x\right)   \otimes
  \psi^\ast \left(q+\tfrac{\epsi}{2}x\right) \,.
  \end{equation*}
  For $\hat \rho^\psi := |\psi\rangle\langle \psi|$ and $A\in S^0(\mathcal{B}(\hfast))$  it holds that
  \[
  \tr_\Hi (\hat \rho^\psi  \,\hat A)\;=\; \int_{\R^{2n}}\df q \df p\,\tr_{\hfast} \left( W^\psi(q,p)\, A(q,p)
  \right)\,.
  \]
  Let  $\nu_\epsi$ be the density of $\lambda_\epsi$ with respect to Lebesgue measure, i.e.\ $  \lambda_\epsi =\nu_\epsi \df q\df p$.
  Then for $\psi_0\in \Pi^\epsi\Hi$ and $\psi(t) := {\rm e}^{-\imag\hat H\frac{t}{\epsi}} \psi_0$ we find for any $a\in S^0(\C)$ that  \begin{align*}
  \tr_\Hi (\hat \rho^{\psi(t)}  \,\hat a)\;&=\;  \tr_\Hi (\hat \rho^{\psi_0}  A(t)) \;=\;  \tr_\Hi (\hat \rho^{\psi_0} \Pi^\epsi A(t)\Pi^\epsi) \\
  &=\;   \tr_\Hi (\hat \rho^{\psi_0} \Pi^\epsi\,\widehat{ a(t)}\Pi^\epsi)  \;+\Or(\epsi^2)
   \;=\;   \tr_\Hi (\hat \rho^{\psi_0} \widehat{ a(t)}) \;+\Or(\epsi^2) \\
   &=\; \int_{\R^{2n}}\df q \df p\,\tr_{\hfast} \left( W^\psi(q,p)\, a\circ\phi^t_\epsi(q,p)
  \right) \;+\Or(\epsi^2)\\
  &=\; \int_{\R^{2n}}\df \lambda_\epsi   \underbrace{ (\nu_\epsi(q,p))^{-1} \,\tr_{\hfast} \left( W^{\psi_0}(q,p)
  \right)}_{=: w_{\pi_0}^{\psi_0}(q,p) }\,a\circ\phi^t_\epsi(q,p) \;+\Or(\epsi^2)\\
   &=\; \int_{\R^{2n}}\df \lambda_\epsi   \,w_{\pi_0}^{\psi_0}\circ\phi^{-t}_\epsi(q,p) \,a(q,p) \;+\Or(\epsi^2)\,,
  \end{align*} 
  where we used that the Liouville measure $\lambda_\epsi$ is invariant under the Hamiltonian flow $\phi^t_\epsi$.
  Thus we define
  the reduced or scalar Wigner transform $w_{\pi_0}^\psi$ associated to the projection $\pi_0$ by
    \begin{equation*}
  w_{\pi_0}^\psi :=   (\nu_\epsi(q,p))^{-1} \tr_{\hfast} W^\psi   = \left( 1-\imag \epsi \tr_{\hfast}\left( \pi_0\{\pi_0,\pi_0\}\right)\right) \tr_{\hfast} W^\psi  +\Or(\epsi^2)\, .
  \end{equation*}

\begin{coro}[\bf Transport of the Wigner function]
For $\psi_0\in \Pi^\epsi\Hi$ 
it holds that 
\[
 w_{\pi_0}^{\psi(t)} \;=\;  w_{\pi_0}^{\psi_0} \circ \phi^{-t}_\epsi + \Or(\epsi^2) 
\]
in the sense that with $\psi(t) := \ez^{-\imag \hat H \frac{t}{\epsi}}\psi_0$ 
\begin{align*}
\langle \psi(t) ,\, \hat a\,\psi(t)\rangle_\Hi \;&=\; 
  \int_{\R^{2n}}\df \lambda_\epsi   \,w_{\pi_0}^{\psi(t)} (q,p) \,a(q,p)
  \\ &=\;
   \int_{\R^{2n}}\df \lambda_\epsi   \,(w_{\pi_0}^{\psi_0}\circ \phi^{-t}_\epsi) (q,p) \,a(q,p) \;+\;\Or(\epsi^2)
\end{align*}
for all $a\in S^0(\C)$.

\end{coro}

\section{The  geometry of the classical phase space}
\label{section:correm}  

In this section we   show that  $\omega_\epsi $ is a symplectic form, compute its Liouville measure $\lambda_\epsi$ and clarify its  relation to the geometry of the line bundle defined by the family of projections $\pi_0(q,p)$ over $\R^{2n}$.

\begin{prop}
 For $\epsi$ small enough $\omega_\epsi$ defined in (\ref{omegadef}) is a symplectic form.
\end{prop}
\begin{proof}
The matrix $\Omega$ is by definition skew-symmetric and bounded, hence
$\mathrm{det}\omega_\epsi\neq 0$ for $\epsi$   small enough. Thus   $\omega_\epsi $ is a  non-degenerate 2-form. To see that $\omega_\epsi$ is closed, note that 
 \begin{align*}
\imag\, \partial_{z_\gamma} \Omega_{\alpha\beta}\; &=\;  \tr_{\hfast}\left(\partial_{z_\gamma}\pi_0[\partial_{z_\alpha}\pi_0,\partial_{z_\beta}\pi_0] \right)\\&\quad + \; \tr_{\hfast}\left(\pi_0[\partial_{z_\gamma}\partial_{z_\alpha}\pi_0,\partial_{z_\beta}\pi_0] \right)+   \tr_{\hfast}\left(\pi_0[\partial_{z_\alpha}\pi_0,\partial_{z_\gamma}\partial_{z_\beta}\pi_0] \right)\\
&=\; \tr_{\hfast}\left(\pi_0[\partial_{z_\gamma}\partial_{z_\alpha}\pi_0,\partial_{z_\beta}\pi_0] \right)-   \tr_{\hfast}\left(\pi_0[\partial_{z_\beta}\partial_{z_\gamma}\pi_0,\partial_{z_\alpha}\pi_0] \right)
 \end{align*}
implies
 $  \partial_{z_\gamma} \Omega_{\alpha \beta} +  \partial_{z_\alpha} \Omega_{\beta\gamma} +  \partial_{z_\beta} \Omega_{\gamma\alpha} =0
 $
 and thus 
 \[
 \df \left( \Omega_{\alpha\beta} \df z^\alpha\wedge \df z^\beta\right)\;=\;   \partial_{z_\gamma} \Omega_{\alpha\beta}\;\df z^\gamma\wedge \df z^\alpha\wedge \df z^\beta=0\,.
 \]
\end{proof}

\begin{prop}
The Liouville measure of the symplectic form $\omega_\epsi$ is given by
\begin{equation}
\lambda_\epsi= \left( 1+\imag \epsi \,\tr_{\hfast}
\left(\pi_0\{\pi_0,\pi_0\}\right)+\mathcal{O}(\epsi^2)\right) \df q^1\wedge\dots\wedge\df p^n. 
\end{equation}
\end{prop}
\begin{proof}
 The Liouville measure is defined by
$
  \lambda_\epsi=\frac{(-1)^{\frac{1}{2}n(n-1)}}{n!}\,\omega_\epsi^{\wedge n}
$.
 With this sign convention $\omega_0^{\wedge n}$ yields the canonically oriented volume form on $\mathbb{R}^{2n}$.
 To compute the first order corrections, it suffices to take all terms into account which contain exactly one term of the form
 $\Omega_{\alpha\beta}\, \df z^\alpha \wedge \df z^\beta$. Due to the antisymmetry of the wedge product, only the off-diagonal terms  
 $\Omega^{qp}_{jj} \, \df q^j \wedge \df p^j$ and $\Omega^{pq}_{jj}\, \df p^j \wedge \df q^j$ contribute and we get  \begin{align*}
 \lambda_\epsi\;&=\;  \left(1+\frac{\epsi}{2} \, \sum_{j=1}^n \left(\Omega^{qp}_{jj} -\Omega^{pq}_{jj} \right)\right) \df q^1\wedge \dots \wedge \df q^n\wedge \df p^1 \wedge \dots \wedge \df p^n  +\Or(\epsi^2)\\
 &=\; \left( 1+\imag \epsi \,\tr_{\hfast}\left(\pi_0\{\pi_0,\pi_0\}\right)\right) \df q^1\wedge\dots\wedge\df p^n+\Or(\epsi^2)\,. 
\end{align*}
\end{proof}

We already mentioned  and it is well known that $\Omega$ is, up to a factor of $\imag$, the curvature of the  Berry connection. For the sake of completeness we explain what  exactly this means.
Originally the Berry phase was introduced  for the adiabatic limit of time-dependent Hamiltonians  by Berry \cite{berry1984quantal}. Shortly thereafter Simon
\cite{simon1983holonomy} realized that this phase could be rewritten as the holonomy of the curvature of a certain line bundle. 
Let us briefly explain this  idea using our notation. The trivial Hilbert bundle
$
 E\;:=\; \R^{2n}  \times \hfast \;\stackrel{P}{\rightarrow} \;\R^{2n} 
$,
where $P$ is the projection onto the first component,   is endowed with a canonical flat connection: for 
$\psi \in\Gamma(E)$ and $X\in \Gamma(T\R^{2n})$ let
$
 (\nabla_X \psi)(z) =X^j \partial_{z_j} \psi (z)  
$.
Within this picture we can regard the symbol $H_0:\R^{2n}\to \mathcal{B}(\hfast)$ as a section in the endomorphism bundle of $E$, i.e.,
$H_0\in \Gamma(\End(E))$ which acts on sections $\psi\in \Gamma(E)$. 
 Associated to an isolated eigenvalue $e_0:\R^{2n} \rightarrow \mathbb{R}$ of $H_0$ is the spectral projection
$\pi_0: \R^{2n} \rightarrow \mathcal{B}(\hfast)$ which again we regard as a section $\pi_0 \in \Gamma(\End(E))$. This projection defines a sub-vectorbundle
$
L := \{ (z,\psi)\,|\, \psi(z) \in \pi_0(z)\hfast\}
$
of $E$.
Now the trivial connection on $E$ induces a connection on $L$ by projection, i.e.\
$ \nabla'_X\phi := \pi_0 \nabla_X \phi$
for $\phi \in \Gamma(L)\subset\Gamma(E)$.  This connection is called the Berry connection.
\begin{prop}\label{berryprop}
 The curvature form $R'$ of the Berry connection $\nabla'$ is   $\frac{1}{2}R'_{ij} \df z^i\wedge \df z^j $, where\begin{equation*}
 R'_{ij}=\pi_0[\partial_{z_i}\pi_0,\partial_{ z_j}\pi_0] 
 \end{equation*}
 and for a line bundle
 \begin{equation*}
 R'_{ij}=\tr_{\hfast}\left(\pi_0[\partial_{z_i}\pi_0,\partial_{ z_j}\pi_0] \right)\,.
 \end{equation*}
\end{prop}
\begin{proof}
 Let $X,Y \in \Gamma (T\R^{2n})$ and $\phi \in \Gamma(L)$, then by definition
 \begin{align*}
 R'(X,Y)\phi \;&:=\;\left(\nabla'_X \nabla'_Y-\nabla'_Y \nabla'_X  -\nabla'_{[X,Y]}\right) \phi   \\
 &\;=\;   \pi_0 \nabla_X \left(\pi_0 \nabla_Y \phi\right)-\pi_0 \nabla_Y\left(\pi_0 \nabla_X\phi\right)  -\pi_0 \nabla_{[X,Y]}\phi
  \\
 &\;=\; X^iY^j  \left(\pi_0 \partial_i \left(\pi_0  {\partial_j}(\pi_0\phi)\right)-\pi_0 {\partial_j}\left(\pi_0  {\partial_i}(\pi_0\phi) \right)  \right)  
  \\
 &\;=\; X^iY^j  \left(\pi_0 [\partial_i,\pi_0]   {\partial_j}(\pi_0\phi)-\pi_0[ {\partial_j},\pi_0]  {\partial_i}(\pi_0\phi)   \right)  \\
 &\;=\; X^iY^j \pi_0 [ \partial_i \pi_0,\partial_j\pi_0] \phi\,.
 \end{align*}
Here we used $\pi_0\phi=\phi$ for $\phi\in \Gamma(L)$ in the third equality, commutativity of partial derivatives in the fourth equality and (\ref{pi0off}) in the fifth equality.
If rank$\pi_0=1$, then
\[
\pi_0 [ \partial_i \pi_0,\partial_j\pi_0] \phi = \pi_0\tr_{\hfast}( \pi_0 [ \partial_i \pi_0,\partial_j\pi_0])\phi =
\tr_{\hfast}( \pi_0 [ \partial_i \pi_0,\partial_j\pi_0])\phi\,. 
\]
\end{proof}
Since $\R^{2n}$ is contractible, there exists a global section $(\varphi_1,\ldots,\varphi_d)$ of the frame bundle associated with $L$, i.e.\ $\varphi_1,\ldots,\varphi_d \in \Gamma(L)$ and $(\varphi_1(z),\ldots,\varphi_d(z))$ is an orthonormal basis of $\pi_0(z)\hfast$ for each $z\in\R^{2n}$.
 With the physicists' bracket notation
the projectors can be written more explicitly in the form $\pi_0(z) =\lvert \varphi_\ell(z)\rangle\langle\varphi_\ell(z)\rvert$ and $R'_{ij}$ can be expressed in terms of this frame.
For the case $d=1$ of a line bundle one finds
\begin{equation*}
R'_{ij}=\langle \partial_i \varphi(z) , \partial_j \varphi(z)\rangle -\langle \partial_j \varphi(z) , \partial_i \varphi(z)\rangle\,.
\end{equation*}
The connection one-form $A$ of the Berry connection with respect to this trivialization is given~by
\begin{equation}
\label{eq:conn}
A=A_j \df z^j=-\imag \,\langle \varphi(z) , \partial_j \varphi(z)\rangle\, \df z^j\,,
\end{equation}
i.e.\ it acts on a local section $\phi = \psi \varphi$ as
\[
\nabla_{\partial_j}' \phi =: \nabla_{\partial_j}'  (\psi\varphi)  =( (\partial_j + \imag A_j)\psi)\varphi\,.
\]
The curvature two-form $\Omega$ is just
\[
\Omega \;=\; \df A \,,\qquad\mbox{i.e.} \qquad \Omega_{ij} \;=\; \partial_i A_j - \partial_jA_i\,.
\]
We remark that the Berry connection is usually discussed in the form of the connection coefficient 
(\ref{eq:conn}) with respect to a local or global trivialization. Note that for our discussion of the semiclassical approximations we used at no point the existence of a trivialization of the bundle $L$. This is why the same reasoning can be used also in the case of Bloch electrons in magnetic fields. There the phase space is the cotangent bundle of a flat torus and the corresponding bundle $L$ is not trivializable, c.f.\ Section~\ref{Hofstadter}.

\section{Discussion of the  literature}
\label{sec-lit}

Semiclassical approximations  for Hamiltonians with operator valued symbols have been discussed many times in the literature and almost all objects appearing in our analysis appeared in some form before.
Hence we give a short roadmap to the literature that we know.

   Our first reference for the corrected symplectic form is a series of papers by
Iida and Kuratsuji \cite{MR819692,iida1987adiabatic,MR930990}. They used path integrals to study adiabatic decoupling of a slow and a fast system. In the path integral formalism they derive an effective Lagrangian for the slow system that incorporates Berry phase effects. After taking the semiclassical limit they are  able to write down their version of the semiclassical equations of motion. However they missed the $M$ term (\ref{Mdef})  and thus their result must be considered incomplete.

In the context of WKB approximations the work of  Bernstein  \cite{bernstein1975geometric} is the earliest reference which contains an additional phase term besides the Berry phase. Kaufman, Ye and Hui \cite{kaufman1987variational} used a variational approach to re-derive those results and expressed them in a form using Poisson brackets. The comparison of these results to the more modern ones is not  straightforward. An attempt to reconcile them has been made by Fukui \cite{PTP.87.927}.

A thorough  discussion of the older literature and the first appearance of the general $M$ term and the modified symplectic form $\omega_\epsi$ is contained in the  work of Littlejohn and Flynn \cite{MR1133847}. 
They  use the Weyl calculus  to diagonalize operators with  matrix-valued symbols. To do this, they must  assume that the spectrum consists only of globally isolated eigenvalues and that the linebundle $L$ defined in the previous section is trivializable. 
They found that - unlike in the scalar case - the equation for the amplitude in the WKB approximation contains  additional phase terms. One of them is the Berry phase and the other the $M$ term, which they call ``no-name term''. The Berry term is gauge dependent, i.e.\ it depends on the choice of a global tirvialization.
This gauge dependence  seemed to be unsatisfactory from their point of view and they searched for a method to express their results in a gauge-independent way,
which turned out to be closely related to the right choice of coordinates. A few lines of calculation show that   the phase space coordinates  $z_j':=z_j-\epsi A_j(z)$ are Darboux or canonical coordinates. Here $A_j$ are the components of the Berry connection (\ref{eq:conn}). This simply means that in these coordinates the symplectic form $\omega_\epsi$ is given by the canonical form $\df q' \wedge \df p'$.
Littlejohn and Flynn on the other hand started out by canonical coordinates and the analysis of gauge invariance led them to consider non-canonical coordinates, which are associated to the symplectic form~$\omega_\epsi$. However it seems to us, as if they failed to realize that their non-canonical coordinates are in fact position and momentum of the slow
degrees of freedom. At least they neither proved nor claimed the statements of our Theorems~\ref{satz:exp} and~\ref{satz:ego}.
 
 The algebraic part of the construction of the Moyal-projection  $\pi$ for $\epsi$-pseudo-differential operators was first done by Helffer and Sj\"ostrand \cite{HS90}. Independently and 
motivated by \cite{MR1133847}, Emmrich and Weinstein \cite{MR1376438} gave a   derivation of the transport equation in WKB approximations which does not rely on eigenfunctions and is intrinsically gauge-independent. 
In their work they give an alternativ construction of the  Moyal-projection $\pi$   and  find the gauge independent classical Hamiltonian $h$ with  the $M$ term in the form (\ref{Malt}). For them the modified symplectic form and the modified flow $\phi^t_\epsi$ play no role. 

Completely independently,  a  formal derivation of the modified semiclassical equations of motion (\ref{eq-motion}) in the special case of Bloch electrons using propagation of wave packets was given by Chang and Niu in \cite{chang-1995-53}.

Based on the approach of \cite{HS90}, the general construction of the super-adiabatic projection $\Pi^\epsi$ was first done by Nenciu and Sordoni \cite{NS:2004}.
Using an adaption of this construction and in addition global trivializations of $L$, in \cite{PST03}   the first rigorous derivation of the modified semiclassical equations for Bloch electrons without strong magnetic fields was given. The fact that such a trivialization is not available in magnetic Bloch bands was the motivation for the present work. 

There have been many more works considering the semiclassical limit of particles with spin, which corresponds to $\hfast = \C^N$. Among them we mention \cite{MR1858263,MR2063266}, where in \cite{MR2063266} an Egorov theorem to any order in $\epsi$ is proved, but not based on a modified classical flow. 

To our best knowledge the statement of Theorem~\ref{satz:exp} on stationary expectations was used only recently by Gat and Avron \cite{gat2003semiclassical} and later in \cite{xiao2005berry}. As far as we understand, there the statement is taken for granted, but not systematically derived or proved.
The statement of Theorem~\ref{satz:ego} was proved in a special case in \cite{PST03}, but otherwise neither claimed nor proved. 

Finally we would like to mention that the physics literature on the semi-classical limit for operator valued symbols is vast and it seems that many observations are reproduced independently several times. For example in \cite{gosselin2007semiclassical} ideas appearing in \cite{MR1133847} are reproduced without citation. Since all these newer works we know of are at most formal, make no   statements similar to ours  and differ from our strategy of proof, we do not list them here.

\section{Application: Magnetic Bloch bands in the Hofstadter model}\label{Hofstadter}

 Consider a gas of noninteracting fermionic  particles on the lattice $\Z^2$ subject to a constant magnetic field $\mathbf{B}=\tiny\left(\begin{array}{cc} 0 & B \\ -B & 0 \end{array}\right)$ with $B\in\R$. The single particle  Hamiltonian is the discrete magnetic Laplacian 
\[
  H^B = \sum_{|n|=1} T^{B}_n  
\] 
acting as a bounded self-adjoint operator on $\ell^2(\Z^2)$. Here $(T^{B}_n \psi)_j = \exp{\frac{\I}{2}   n\cdot \mathbf{B} j } \psi_{j-n}$ for $n\in \Z^2$   are magnetic translations. This is a simple model for   particles in a periodic potential in the tight binding approximation.  For $B=0$  the operator $H^{B=0}$ is invariant under lattice translations and thus diagonalized by the Fourier transformation  $\mathcal{F}:\ell^2(\Z^2)\to L^2(\R^2/(2\pi \Z^2))$, where $\mathcal{F} H^{B=0}\mathcal{F}^*= E$ becomes the multiplication operator with the function $E(k)=2(\cos(k_1)+\cos(k_2))$. 

For  $B\not=0$, the operator  $H^B$ is invariant under 
the dual magnetic translations $(\tilde T^{B}_n \psi)_j = \exp{-\frac{\I}{2}   n\cdot \mathbf{B} j } \psi_{j-n}$,   which, in general,  do not form a unitary group representation, since $ \tilde T^{B}_{(1,0)}   \tilde T^{B}_{(0,1)} = \exp{\I B} \,\tilde T^{B}_{(0,1)}  \tilde T^{B}_{(1,0)} $.
However, for magnetic fields $B=B_0 = 2\pi \,\frac{p}{q}$ that are rational multiples of $2\pi$, the magnetic field and the lattice are commensurable in the sense that the magnetic flux per unit cell is a rational multiple of the flux-quantum $\Phi_0 = \frac{2\pi\hbar}{e} =2 \pi$.  As a consequence one can extend the dual magnetic translations 
to a unitary representation of the subgroup $\Gamma_q := \{ \gamma\in\Z^2\,|\, \gamma_1\in q\Z\}$ of $\Z^2$
by defining
\[
\tilde T^{B_0}_\gamma := \left( \tilde T^{B_0}_{(1,0)} \right)^{\gamma_1}   \left( \tilde T^{B_0}_{(0,1)} \right)^{\gamma_2}
\]
and thus $H^{B_0}$ can be fibered with respect to this group representation:  
Let $M_q := [0,2\pi/q) \times [0,2\pi)$ be the reduced Brillouin zone and define the magnetic  Bloch-Floquet transformation as
\begin{align*}
U^{B_0} : &\; \ell^2(\Z^2) \to  L^2(M_q;\C^q)\,,\\ &\;(\psi_j)_{j\in\Z^2} \mapsto (U^{B_0}\psi)_m(k) := 
\exp{-\I k \cdot m} \sum_{\gamma\in\Gamma_q} \exp{\I k \cdot \gamma}  (\tilde T^{B_0}_\gamma \psi)_m \mbox{ for } m = 0,\ldots, q-1\,.
\end{align*}
Then $U^{B_0}$ is unitary and  because of $(U^{B_0}\psi)_m(k + \gamma^*) = \exp{-\I \gamma^* \cdot m} (U^{B_0}\psi)_m(k)$ for any $\gamma^*$ in the dual lattice   $\Gamma^*_q = := \{  (\gamma^*_1,\gamma^*_2) \in\R^2\,|\,
\gamma^*_1\in \frac{2\pi}{q} \Z \,,\,\gamma^*_2\in 2\pi\Z  \}$ the functions in its range can be naturally extended to
$\tau$-equivariant functions in $L^2_{\rm loc}(\R^2_k; \C^q)$. Here a function $f:\R^2\to \C^q$ is said to be $\tau$-equivariant if it satisfies
\[
f( k + \gamma^*)  =  \tau(\gamma^*) f(k) := {\rm diag}( 1, \exp{-\I \gamma^*_1  },\ldots,  \exp{-\I (q-1)\gamma^*_1  })  f(k)\quad\mbox{for all } \gamma^*\in \Gamma^*_q
\]
for the unitary representation $\tau$ of $\Gamma^*_q$.

The image $U^{B_0} H^{B_0} U^{B_0 *}$ of the Hamiltonian $H^{B_0}$ acts on  $\Hi = L^2(M_q;\C^q)$ as the matrix-valued multiplication operator 
\[
\small
H_0(k) = \left(\begin{array}{ccccc}
2\cos(k_2)  & \exp{-\I k_1}  & 0 &   \cdots & \exp{ \I k_1}\\
\exp{ \I k_1} & 2\cos(k_2 + B_0) & \exp{-\I k_1}&   \cdots& 0 \\
0 & \exp{ \I k_1} &  2\cos(k_2 + 2 B_0) &      \cdots &0\\
\vdots &     \ddots &  \ddots & \ddots &0\\
0 &&&&\exp{-\I k_1} \\
\exp{-\I k_1} & 0& \cdots  & \exp{ \I k_1}& 2\cos(k_2 + (q-1) B_0)
\end{array}
\right)\,,
\]
which is again $\tau$-equivariant, i.e.\ $H_0(k+\gamma^*) = \tau(\gamma^*) H_0(k) \tau(-\gamma^*)$ and thus $\Gamma^*_q$-periodic up to unitary equivalence. It is well known that for $q$ odd, $H_0(k)$ has $q$ distinct eigenvalues  $e^{(1)}(k)<\ldots< e^{(q)}(k)$ such that $e^{(j)}(M_q)\cap e^{(i)}(M_q)=\emptyset$ for $j\not= i$. By the $\tau$-equivariance and analyticity of $H_0(k)$, $e^{(j)}(k)$ are $\Gamma^*_q$-periodic real analytic functions. For $q$ even, all bands but the middle ones are isolated, where $e^{(q/2)}(k_*) = e^{(q/2+1)}(k_*)$ for one point $k_*\in M_q$. Thus except for the middle bands the same conclusions hold. The spectrum of $H^{B_0}$ thus consists of $q$ resp.\ $q-1$ disjoint closed intervals and if plotted as a function of $B_0$ gives rise to the famous Hofstadter butterfly with its fractal   structure. 
 
We now show that small (also irrational)  perturbations of $B_0$ give rise to a semiclassical problem for a Hamiltonian with operator valued symbol and that we can use the semiclassical approximations developed in the previous chapters in order to understand $H^B$ for a magnetic field  $B= B_0 + \epsi b$ in terms of a classical system based on the geometry of $H_0(k)$. 
To see this note that for $B= B_0 + \epsi b$ and in the presence of an electric potential $V^\epsi(n)= V(\epsi n)$ with $V(x) = V_{\rm b}(x) + \mathcal{E}\cdot x$ for  $V_{\rm b}\in C^\infty_{\rm b}(\R^2)$,  the Bloch-Floquet  transformation with respect to $B_0$  yields the Hamiltonian
 \begin{eqnarray*}
 H^\epsi &:=& U^{B_0} H^{B } U^{B_0 *} = H_0(k_1 +{\textstyle \frac{1}{2} } b \,\I\epsi \partial_{k_2}^\tau,\, k_2 -{\textstyle \frac{1}{2} } b \,\I\epsi \partial_{k_1}^\tau ) + V(\I\epsi\nabla_k^\tau)\\& =& H_0(k + {\textstyle \frac{1}{2} } {\mathbf{b}} \,\I\epsi\nabla_k^\tau  )+ V(\I\epsi\nabla_k^\tau)
 \end{eqnarray*}
 acting on $ L^2(M_q;\C^q)$. Here $\partial_{k_j}^\tau$ denotes the derivatives with $\tau$-equivariant boundary conditions and the matrix entries, e.g.\ $2\cos(k_2 -{\textstyle \frac{1}{2} } b \,\I\epsi \partial_{k_1}^\tau)$, 
 are defined by the functional calculus for self-adjoint operators.
 However, the operator $H^\epsi$ can also be understood as the $\epsi$-Weyl-quantization of  the matrix-valued function
\[
 H(k, r) := H_0( k_1 + {\textstyle \frac{1}{2} } b r_2, k_2 - {\textstyle \frac{1}{2} } b r_1) + V(r)\,,
 \]
i.e.\  $H^\epsi =   H(k, \I\epsi\nabla_k^\tau)$. In order to apply the Weyl-calculus from $\R^2$ here, one identifies
$L^2(M_q,\C^q)$ with the space 
$L^2_\tau := \{ f\in L^2_{\rm loc}(\R^2)\,|\, f\mbox{ is $\tau$-equivariant}\}$ with the norm $\|f\|_\tau^2 = \frac{q}{(2\pi)^2}\int_{M_q}|f(k)|^2\D k$  and restricts to $\tau$-equivariant symbols, see \cite{PST03} or Appendix~B in \cite{Teufel:2003} for details. With this modification all the results of this paper hold with the identical proofs for the Hamiltonian with symbol $H(k,r)$ and the underlying classical system on the phase space $\mathbb{T}^2\times \R^2$, which we describe next.

Let $e^{(j)}(k)$ be an isolated band of $H_0(k)$ with spectral projection $\pi^{(j)}(k)$. Then
$\tilde e^{(j)} (k,r) := e^{(j)}(k +\frac12 \mathbf{b} r) + V(r)$ is an isolated band for $H(k,r)$ with spectral projection 
$\tilde \pi^{(j)}  (k,r) = \pi^{(j)} (k +\frac12 \mathbf{b} r)$. A simple computation shows that the associated classical system
is
\[
h^{(j)} (k,r) = e^{(j)}(k +\tfrac12 \mathbf{b} r) + V(r) + \epsi  b \mathcal{M}^{(j)}( k +\tfrac12 \mathbf{b} r) \,,
\]
where
\[
 \mathcal{M}^{(j)}  =  {\rm Im} \;\tr_{\C^q}\left(\pi^{(j)}  \,\partial_{k_1}\pi^{(j)} \,(H_0 - e^{(j)}) \,\partial_{k_2}\pi^{(j)} 
 \right)
\]
is the effective magnetic moment of the particles in the $j$th band. With
\[
\Omega^{(j)}    := -\I \,\tr_{\C^q} ( \pi^{(j)}  [\partial_{k_1} \pi^{(j)} , \partial_{k_2} \pi^{(j)} ])  
\] 
being the curvature of the Berry connection on the $j$th eigenspace bundle of $H_0(k)$ one finds
\[
\Omega^{(j)kk}  (k,r) = \left(\begin{array}{cc} 0 &  \Omega^{(j)} (k +\tfrac12 \mathbf{b} r)\\
 -\Omega^{(j)} (k +\tfrac12 \mathbf{b} r)&0\end{array}\right) =: \mathbf{\Omega}^{(j)}(k +\tfrac12 \mathbf{b} r)\,,
\]
$  \Omega^{(j)rk} = -\tfrac12 \mathbf{b \Omega}^{(j)} =  \frac12 b\Omega^{(j)} E_2 $ and $  \Omega^{(j)rr} =- \tfrac14 \mathbf{b \Omega}^{(j) }\mathbf{b}= \frac14  b\Omega^{(j)}\mathbf{b}$.
The corresponding equations of motion (\ref{eq-motion}) become much simpler if one 
changes coordinates to the kinetic momentum $\kappa := k +\tfrac12 \mathbf{b} r$, leading to   
\[
h^{(j)} ( r,\kappa) = e^{(j)}(\kappa) + V(r) + \epsi  b \mathcal{M}^{(j)}( \kappa) \quad \mbox{and}\quad 
\omega^{(j)}_\epsi(\kappa) =   \left(\begin{array}{cc} - \mathbf{b}   \;&\;  E_2 \\[1mm]
-   E_2 \;&\;  \epsi \mathbf{\Omega}^{(j)}(\kappa)
\end{array} \right) 
\]
 with Liouville measure  $\D \lambda_\epsi^{(j)} =  (1 + \epsi b \Omega^{(j)}(\kappa))\,\D r\D \kappa$.
The Hamiltonian equations of motion are thus
\begin{eqnarray*}
\left(\begin{array}{c} \dot r \\ \dot \kappa \end{array}\right) &=& \frac{1}{(1 + \epsi b \Omega^{(j)}(\kappa))}   
 \left(\begin{array}{cc} -\epsi\mathbf{\Omega}^{(j)}(\kappa) \;&\;  E_2 \\[1mm]
-  E_2 \;&\;    \mathbf{b}
\end{array} \right)\left(\begin{array}{c} \partial_r h(\kappa,r) \\[1.5mm] \partial_\kappa h(\kappa,r) \end{array}\right)
\\&=&  \frac{1}{(1 + \epsi b \Omega^{(j)}(\kappa))} 
\left(\begin{array}{c}  \partial_\kappa \left(e^{(j)}(\kappa)   + \epsi  b \mathcal{M}^{(j)}( \kappa)\right)  - \epsi\mathbf{\Omega}^{(j)}(\kappa) \,\partial_r V(r) \\[1.5mm]  -\partial_r V(r) + \mathbf{b}  \,\partial_\kappa \left(e^{(j)}(\kappa)   + \epsi  b \mathcal{M}^{(j)}( \kappa)\right) \end{array}\right)
 \,,
\end{eqnarray*}
or, equivalently, by
\[
\dot r =  \partial_\kappa \left(e^{(j)}(\kappa)   + \epsi  b \mathcal{M}^{(j)}( \kappa)\right)    - \epsi\mathbf{\Omega}^{(j)}(\kappa)\,\dot \kappa\,,\qquad \dot \kappa =   -\partial_r V(r) + \mathbf{b} \,\dot r\,.
\]  
Finally note that by $\tau$-equivariance of $H_0(k)$  the functions $e^{(j)}$,  $\mathcal{M}^{(j)}$ and $\Omega^{(j)}$ are $\Gamma^*_q$-periodic. However, due to the symmetry of the original problem   they are even periodic with respect to $\frac{2\pi}{q} \Z^2$. Therefore the classical system can be defined on the phase space $\mathbb{T}_q\times \R^2$ where $\mathbb{T}_q$ is the torus $[0,2\pi/q)^2$.

We now compute the free energy per unit area of a gas of noninteracting Hofstadter particles at temperature $T=1/\beta$ and chemical potential $\mu$, also called the pressure.
 To this end let $\chi_n:\R^2 \to [0,1]$ be a sequence of ``smooth characteristic functions'' on $\Lambda_n := [-n,n]^2$, i.e.\ $\chi_n\in S^0(\R)$ uniformly with ${\rm supp} \chi_n \in \Lambda_n$ and $\chi_n|_{\Lambda_{n-1}} \equiv 1$. 
 With a slight abuse of notation we set $|\Lambda_n| := \|\chi_n\|_{L^1}$. 
 \begin{theo}[The thermodynamic pressure in the Hofstadter model]\label{pressure}
Let $B_0 = 2\pi \,\frac{p}{q}$ with $q$ odd and let $B= B_0 + \epsi b$ and $V(r)=0$.  Then for $0< \beta<\infty$, $\mu\in\R$ and $\epsi>0$ small enough, the limit
\[
p(B,\beta,\mu) = \beta^{-1} \lim_{n\to\infty}  \,\frac{1}{\frac{1}{ \epsi  ^2}|\Lambda_n|} \,\tr_{\ell^2(\Z^2)}\left( \chi_n(\epsi   x) \ln\left(
1 + \exp{-\beta (H^B-\mu)}
\right)\right)
\] 
exists and is approximated by the classical expression
\[
p(B,\beta,\mu) = \frac{q}{(2\pi)^2} \sum_{j=1}^q \int_{\mathbb{T}_q}   \left(1 + \epsi b \Omega^{(j)}(\kappa)\right) 
 \ln\left(
1 + \exp{-\beta ( h^{(j)}(\kappa) - \mu)}
\right)\D \kappa
+ \Or(\epsi  ^2)\,.
\]
Note that $\frac{1}{ \epsi  ^2}|\Lambda_n|$ is the area of the region in space to which $\chi_n(\epsi   x)$ restricts. 
\end{theo}
\begin{proof} The existence of the limit   follows immediately from invariance of the Hamiltonian under magnetic translations. By unitarity of the Bloch-Floquet transformation  and Theorem~\ref{satz:exp} we find
\begin{eqnarray*}\lefteqn{
\tr_{\ell^2(\Z^2)}\left( \chi_n(\epsi   x) \ln\left(
1 + \exp{-\beta (H^B-\mu)}
\right)\right)
=
\tr_{\Hi}\left( \chi_n(\I \epsi   \nabla_k^\tau) \ln\left(
1 + \exp{-\beta ( H^\epsi -\mu)}
\right)\right)}\\
&=& \sum_{j=1}^q \tr_{\Hi}\left( \Pi^{(j)} \chi_n(\I \epsi   \nabla_k^\tau) \ln\left(
1 + \exp{-\beta ( H^\epsi -\mu)}
\right)\right)\\
&=& \sum_{j=1}^q  \frac{1}{(2\pi\epsi)^2}\int_{M_q\times \R^2} \left(1 + \epsi b \Omega^{(j)}(\kappa)\right)
 \chi_n( r)\, \ln\left(
1 + \exp{-\beta ( h^{(j)}(\kappa) - \mu)}
\right)\D \kappa\,\D r + \Or(  \|\chi_n\|_{L^1})\\
&=& \sum_{j=1}^q  \frac{|\Lambda_n|}{(2\pi\epsi  )^2}\int_{M_q } \left(1 + \epsi b \Omega^{(j)}(\kappa)\right)
   \ln\left(
1 + \exp{-\beta ( h^{(j)}(\kappa) - \mu)}
\right)\D \kappa  + \Or( |\Lambda_n|)\,.
\end{eqnarray*}
Now we can restrict the integration to $\mathbb{T}_q$ to get a factor of $q$, divide by the area $ \epsi^{-2}|\Lambda_n|$ and take the limit to prove the second statement.
\end{proof}

The thermodynamic pressure is the starting point for computing several physically relevant quantities like
the magnetization $M(B,\beta,\mu) = \partial_B p(B,\beta,\mu)$ and the  density $\rho(B,\beta,\mu) = \partial_\mu  p(B,\beta,\mu)$ and derivatives thereof. 
\begin{coro}[The magnetization of the Hofstadter model at rational $B$]\label{MagCor}
Let $B_0 = 2\pi \,\frac{p}{q}$ with $q$ odd   and $V(r)=0$.  Then for $0< \beta<\infty$, $\mu\in\R$  
it holds that 
\[
M(B_0 , \beta,\mu)   =  \frac{q}{(2\pi)^2} \sum_{j=1}^q \int_{\mathbb{T}_q}   \left( f_{\beta,\mu}( h^{(j)} ) \mathcal{M}^{(j)}   + \tfrac{1}{\beta}
 \ln\left(
1 + \exp{-\beta ( h^{(j)}  - \mu)}
\right) \Omega^{(j)} \right)\D \kappa
 \,,
\]
where $f_{\beta,\mu}(E) := (1 + \exp{\beta(E-\mu)})^{-1}$ is the Fermi-Dirac distribution. 
\end{coro}
\begin{proof}
The formula is obtained by taking the derivative of the semi-classical expression for $p$ from Theorem~\ref{pressure}.
Noting that $\partial_B = b^{-1} \partial_\epsi$ shows that the error term does not contribute to the derivative,   since it is $\Or(\epsi^2)$.
\end{proof}
This formula for the magnetization in solids is rather recent even on a heuristic level, see \cite{gat2003semiclassical,xiao2005berry}. It was proved by a completely different approach as a special case of a more general formula in \cite{SBT12}. Note however, that at least on a formal level the present derivation using the semiclassical model is much simpler.

As a final application of the Egorov theorem or, more precisely, of   Proposition~\ref{commprop}, we compute the current density in the Hofstadter model at zero temperature when a constant  electric field $\mathcal{E}\in\R^2$ is applied and when  $\mu$ lies  in a gap between $e^{(m)}$ and $e^{(m+1)}$.
The corresponding state is given by $\Pi^\epsi = \sum_{j=1}^m \Pi^{(j)}$ and its localization to the region $\epsi^{-1}\Lambda_n$ by $\Pi^\epsi_n := \Pi^\epsi \chi_n(\epsi x) \Pi^\epsi$.

\begin{theo}[Current density in the Hofstadter model] 
Let $B_0 = 2\pi \,\frac{p}{q}$ with $q$ odd  and let $B= B_0 + \epsi b$ and $V(r)=\mathcal{E}\cdot r$. Let $\mu\in\R$   be such that $e^{(m)}(k) < \mu < e^{(m+1)}(k)$  for one $m\in\{1,\ldots, q\}$ and all $k\in \mathbb{T}_q$. 
Then 
with $\Pi^\epsi = \sum_{j=1}^m \Pi^{(j)}$ and $\Pi^\epsi_n := \Pi^\epsi \chi_n(\epsi x) \Pi^\epsi$ it holds that
\[
j(B,\mathcal{E}) :=  \lim_{n\to\infty}  \,\frac{1}{\frac{1}{ \epsi  ^2}|\Lambda_n|} \,\tr_{\ell^2(\Z^2)}\left( \Pi^\epsi_n\,J^\epsi  \right) = \mathcal{E}^\perp \,\tfrac{q}{(2\pi)^2} \sum_{j=1}^m \int_{\mathbb{T}_q} \Omega^{(j)}(\kappa)\D\kappa + \Or(\epsi)\,,
\]
where $J^\epsi := \frac{\I}{\epsi}[ H^B, x]$ is the current operator. 
\end{theo}
As $\frac{q}{2\pi} \int_{\mathbb{T}_q} \Omega^{(j)}(\kappa)\D\kappa\in \Z$ is the Chern number of the $j$th magnetic sub-band at magnetic field $B_0$, the Hall current is quantized and independent of $B$ in a neighborhood of $B_0$ at leading order.
 Since this result is basically well known and proven in greater generality by other means, e.g.\ \cite{BSB}, we  only give the simple computation based on the semiclassical model, but skip the mathematical details. However, the following can be easily made rigorous by using the results of the previous sections. On the other hand we believe that the strength of the semiclassical model lies in the fact that its application is very simple and leads in a straight-forward way to the correct physical expressions.
 
According to (\ref{CommuPoisson})  we have that  $J = \frac{\I}{\epsi} (H\# r- r\# H)\in S^0(\R)$ satisfies
\[
\Pi^{(j)} \,\hat J_{l} \Pi^{(j)} = \Pi^{(j)} \,\hat X_{h,{r_l}}\Pi^{(j)}  + \Or(\epsi^2)
\]
and thus 
 \begin{eqnarray*}\lefteqn{
\sum_{j=1}^m\tr_{\Hi}\left( \hat \chi_n  \Pi^{(j)} \,\hat J_{l} \Pi^{(j)} \right)
= \sum_{j=1}^m\tr_{\Hi}\left( \hat \chi_n  \Pi^{(j)} \,\hat X_{h,{r_l}}\Pi^{(j)} \right)+ \Or(\epsi^2|\Lambda_n|)}\\&=& \tfrac{q}{(2\pi\epsi)^2}\sum_{j=1}^m \,\int\limits_{\mathbb{T}_q\times\R^2}\hspace{-8pt}\D \lambda^{(j)}_\epsi
 \frac{\chi_n(r)}{(1 + \epsi b \Omega^{(j)}(\kappa))} 
\left(  \partial_\kappa \left(e^{(j)}(\kappa)   + \epsi  b \mathcal{M}^{(j)}( \kappa)\right)  - \epsi\mathbf{\Omega}^{(j)}(\kappa) \,\partial_r V(r)\right)\,.
\end{eqnarray*}
Since $J^\epsi = \frac{1}{\epsi}\hat J$, dividing by $\frac{1}{ \epsi  ^2}|\Lambda_n|$ and taking the thermodynamic limit yields
\begin{eqnarray*}
j(B,\mathcal{E})& =& \tfrac{q}{\epsi (2\pi)^2}\sum_{j=1}^m \,\int_{\mathbb{T}_q}\D \kappa
\left(  \partial_\kappa \left(e^{(j)}(\kappa)   + \epsi  b \mathcal{M}^{(j)}( \kappa)\right)  - \epsi\mathbf{\Omega}^{(j)}(\kappa) \,\mathcal{E}\right)\\
&=&  \mathcal{E}^\perp \,\tfrac{q}{(2\pi)^2} \sum_{j=1}^m \int_{\mathbb{T}_q} \Omega^{(j)}(\kappa)\D\kappa + \Or(\epsi)\,.
\end{eqnarray*}
Here we used that the leading term is the integral of a gradient over a manifold without boundary and thus zero by Stokes' theorem. 
This is the well known text book statement that filled bands cannot contribute to the current. The standard semiclassical model without the first order corrections would imply exactly this.

\appendix

\section{Technical lemmas}

In the following proofs we apply the Helffer-Sj\"ostrand formula for functions in $\mathcal{A}$ as developed in \cite{Davies}.

\begin{lemma}\label{lem:est}
Let $A$ be a self-adjoint operator on $\Hi$ with domain $D$ and $P\in\mathcal{B}(\Hi)$ an orthogonal projection such that $PD\subset D$ and 
\[
\| [ A,P ] \| \leq d\,.
\]
Then   for any $f\in \mathcal{A}$ there is a constant $C <\infty$ such that
\begin{equation}\label{esti1}
\| [ f(A), P] \| \;\leq \; C \, d
\end{equation}
and
\begin{equation}\label{esti2}
\|  P f(A) P -   f(PAP)   \|  \;\leq\; C \,d^2\,.
\end{equation}
\end{lemma}
\begin{proof} The bound (\ref{esti1}) is an immediate consequence of the Helffer-Sj\"ostrand formula. For 
(\ref{esti2})
first note that $PAP$ is self-adjoint on $\Hi$ with domain $PD\oplus P^\perp\Hi$.
Hence $f(PAP) \in \mathcal{B}( \Hi)$ is well-defined by the functional calculus and it holds that
$f(PAP)=Pf(PAP)P$.
We next show that
\begin{equation}\label{reso}
\| P \left( (A-z)^{-1} - (PAP-z)^{-1}\right) P\|  \;\leq\; \frac{d^2}{|{\rm Im}(z)|^3}\,.
\end{equation}
This follows from taking norms in the identity 
\begin{align*}\lefteqn{
 P \left( (A-z)^{-1} - (PAP-z)^{-1}\right) P= P  (A-z)^{-1} P  -P (PAP-z)^{-1}  P}
 \\&\quad=\;
 P (A-z)^{-1} P \left( (PAP-z)  - (A-z) \right)P(PAP-z)^{-1} P\\
 &\quad\quad +  P (A-z)^{-1} P  (A-z)  P(PAP-z)^{-1} P -P (PAP-z)^{-1}  P\\
 &\quad=\; P (A-z)^{-1} [P,  A]  P(PAP-z)^{-1} P\\
 &\quad=\; P (A-z)^{-1} (1-P) [P,  A]  P(PAP-z)^{-1} P\\
 &\quad=\; P [P, (A-z)^{-1}] [P,  A]  P(PAP-z)^{-1} P\\
 &\quad=\; P(A-z)^{-1}[A,P] (A-z)^{-1}[P,A] P(PAP-z)^{-1} P\,.
\end{align*}
Now let $\tilde f$ be an almost analytic extension of $f$ with $|\partial_{\bar z} \tilde f(z)|\leq c |{\rm Im}(z)|^3 $.
By the Helffer-Sj\"ostrand formula and  (\ref{reso})
\begin{align*} \lefteqn{
 \| P f(A) P - P f(PAP) P\| \;=}\\&\quad=\; \left\|   \tfrac{1}{\pi} \int_\C\partial_{\bar z} \tilde f(z)\,P \left((A-z)^{-1} - (PAP-z)^{-1}
  \right)P\,\D x\D y\right\|\\
  &\quad\leq \;   \tfrac{1}{\pi} \int_\C |\partial_{\bar z} \tilde f(z)|\left\|P \left((A-z)^{-1} - (PAP-z)^{-1}
  \right)P\right\| \D x\D y \;\leq\; C_f \,d^2\,.
\end{align*}%
 \end{proof}

\begin{lemma}\label{lem:est2}
Let $A,B$ be   self-adjoint operators on $\Hi$ with domain $D$ 
such that $A-B$ extends to a bounded operator with $\| A-B\|\leq d$. 
Then for any $f\in \mathcal{A}$ there is a constant $C <\infty$ such that
\begin{equation}\label{esti3}
\|    f(A)   -   f(B)   \|  \;\leq\; C \,d \,.
\end{equation}
\end{lemma}
\begin{proof} 
Clearly 
\begin{equation}\label{reso2}
\|   (A-z)^{-1} - (B-z)^{-1} \| \;=\;  \|   (A-z)^{-1} (B-A) (B-z)^{-1} \|    \;\leq\; \frac{d }{|{\rm Im}(z)|^2}\,.
\end{equation}
Let $\tilde f$ be an almost analytic extension of $f$ with $|\partial_{\bar z} \tilde f(z)|\leq c |{\rm Im}(z)|^2 $.
By The Helffer-Sj\"ostrand formula and  (\ref{reso2})
\begin{align*} 
 \|   f(A)   -   f(B)  \|\; &=\; \left\|   \frac{1}{\pi} \int_\C\partial_{\bar z} \tilde f(z)  \left( (A-z)^{-1} - (B-z)^{-1}\right) \D x\D y\right\|\\
  &\leq \;     \frac{1}{\pi} \int_\C |\partial_{\bar z} \tilde f(z)|\left\| (A-z)^{-1} - (B-z)^{-1}\right\| \D x\D y\;\leq\; C_f \,d \,.
\end{align*}%
\end{proof}

\begin{lemma}\label{lem-func}
Let $h\in S^0(\epsi,\C)$ be real-valued and $f\in \mathcal{A}$. Then 
$\hat h$ is a bounded self-adjoint operator and 
\begin{equation}\label{1}
\| f(\hat h) - \widehat{f(h)}\| = \Or(\epsi^2)\,.
\end{equation}
\end{lemma}
\begin{proof}
The Weyl quantization of a real valued symbol in $S^0(\epsi,\C)$  is bounded and symmetric and thus self-adjoint. Following Exercise~24 of Chapter~2 in \cite{martinez2002introduction}, it suffices to show (\ref{1}) for the resolvents at $z\in\C\setminus\R$:
\begin{align*}
\lefteqn{ \hspace{-40pt}
(\hat h - z)^{-1} - {\rm Op}^{\rm W}( (h-z)^{-1} )=(\hat h - z)^{-1} \big(
1 - (\hat h -z )\, {\rm Op}^{\rm W}( (h-z)^{-1})
\big)
}\\
&\quad=\; (\hat h - z)^{-1} \big(
1 - {\rm Op}^{\rm W}( (  h -z )\,\#\, (h-z)^{-1})
\big)\\
&\quad=\; (\hat h - z)^{-1} \big( \tfrac{\imag\epsi}{2} \underbrace{\{ (  h -z ) , (h-z)^{-1}\}}_{=0} \,+\,\Or(\epsi^2)
\big)\,. 
\end{align*}%
\end{proof}

\bibliographystyle{alpha}

\end{document}